\documentclass[reqno, 11pt, a4paper]{amsart}
\usepackage{fullpage}  
\usepackage{color,calc,graphicx,amsfonts,amsthm,amscd,epsfig,psfrag,amsmath,amssymb,enumerate,dsfont}
\usepackage{caption}
\usepackage{subcaption}
\usepackage{epstopdf}

\usepackage{pdfsync}
\usepackage{stmaryrd} 
\usepackage{marginnote}
\usepackage[top=2.5cm, bottom=2.5cm, outer=2.5cm, inner=2.5cm, heightrounded, marginparwidth=1.9cm, marginparsep=0.3cm]{geometry}

\usepackage[colorlinks=true, pdfstartview=FitV, linkcolor=blue,citecolor=blue, urlcolor=blue,pagebackref=false]{hyperref}
\usepackage{color}
\definecolor{light}{gray}{.9}
\usepackage[showlabels,sections,floats,textmath,displaymath]{}
\usepackage{booktabs}

\usepackage{appendix}
\usepackage[percent]{overpic}
\usepackage[numeric,initials,nobysame,citation-order]{amsrefs}

\newcommand{\beq}{\begin{equation}}
\newcommand{\eeq}{\end{equation}}
\newcommand{\ie}{\hbox{\it i.e.\ }}

\reversemarginpar
\newlength\fullwidth
\numberwithin{equation}{section}

\DeclareMathSymbol{\leqslant}{\mathalpha}{AMSa}{"36} 
\DeclareMathSymbol{\geqslant}{\mathalpha}{AMSa}{"3E} 
\DeclareMathSymbol{\eset}{\mathalpha}{AMSb}{"3F}     
\renewcommand{\leq}{\;\leqslant\;}                   
\renewcommand{\geq}{\;\geqslant\;}                   

\newcommand{\1}{\mathds{1}}

\renewcommand{\l}{\lambda}
\renewcommand{\L}{\Lambda}

\renewcommand{\l}{\lambda}
\renewcommand{\a}{\alpha}

\newcommand{\z}{\zeta}

\renewcommand{\O}{\Omega}

\newtheorem{theorem}{Theorem}[section]

\newtheorem{lemma}[theorem]{Lemma}
\newtheorem{proposition}[theorem]{Proposition}

\newtheorem{claim}[theorem]{Claim}
\newtheorem{definition}[theorem]{Definition}

\newcommand{\N}{\mathbb N}

\newcommand{\bbN}{{\ensuremath{\mathbb N}} }

\newcommand{\bbP}{{\ensuremath{\mathbb P}} }

\newcommand{\bbR}{{\ensuremath{\mathbb R}} }

\let\a=\alpha      
  \let\h=\eta      \let\l=\lambda
          \let\p=\pi  
      
 \let\x=\xi \let\z=\zeta
   \let\G=\Gamma  \let\L=\Lambda 
\let\O=\Omega

\begin{document}
\title{Monotonicity and Condensation in Homogeneous Stochastic Particle Systems}
\author{Thomas Rafferty}
\address{T.Rafferty. Centre for Complexity Science, University of Warwick, Coventry, CV4 7AL, UK}
\email{t.rafferty@warwick.ac.uk}
\author{Paul Chleboun}
\address{P.Chleboun. Mathematics Institute and Centre for Complexity Science, University of Warwick, Coventry, CV4 7AL, UK}
\email{p.i.chleboun@warwick.ac.uk}
\author{Stefan Grosskinsky}
\address{S.Grosskinsky. Mathematics Institute and Centre for Complexity Science, University of Warwick, Coventry, CV4 7AL, UK}
\email{s.w.grosskinsky@warwick.ac.uk}
\date{\today}

\begin{abstract}
We study stochastic particle systems that conserve the particle density and exhibit a condensation transition due to particle interactions. We restrict our analysis to spatially homogeneous systems on finite lattices with stationary product measures, which includes previously studied zero-range or misanthrope processes. All known examples of such condensing processes are non-monotone, i.e. the dynamics do not preserve a partial ordering of the state space and the canonical measures (with a fixed number of particles) are not monotonically ordered. For our main result we prove that condensing homogeneous particle systems with finite critical density are necessarily non-monotone. On finite lattices condensation can occur even when the critical density is infinite, in this case we give an example of a condensing process that numerical evidence suggests is monotone, and give a partial proof of its monotonicity.
\end{abstract}
\maketitle

\section{Introduction}

We consider stochastic particle systems which are probabilistic models describing transport of a conserved quantity on discrete geometries or lattices. 
Many well known examples are introduced in \cite{spitzer70}, including zero-range processes and exclusion processes, which are both special cases of the more general family of misanthrope processes introduced in \cite{misanthrope}. 
We focus on spatially homogeneous models with stationary product measures and without exclusion restriction, which can exhibit a condensation transition that has recently been studied intensively.

A condensation transition occurs when the particle density exceeds a critical value and the system phase separates into a fluid phase and a condensate. The fluid phase is distributed according to the maximal invariant measure at the critical density, and the excess mass concentrates on a single lattice site, called the condensate. Most results on condensation so far focus on zero-range or more general misanthrope processes in thermodynamic limits, where the lattice size and the number of particles diverge simultaneously. Initial results are contained in \cite{drouffe98,Jeon2000,evans00}, and for summaries of recent progress in the probability and theoretical physics literature see e.g.\ \cite{Godreche2012,Chleboun2013a,Evans2014}. Condensation has also been shown to occur for processes on finite lattices in the limit of infinite density, where the tails of the single site marginals of the stationary product measures behave like a power law \cite{ferrarietal07}. In general, condensation results from a heavy tail of the maximal invariant measure \cite{stefan}, and so far most studies focus on power law and stretched exponential tails \cite{armendarizetal09,agl}. As a first result, we generalize the work in \cite{ferrarietal07} and provide a characterization of condensation on finite lattices in terms of the class of sub-exponential tails that has been well studied in the probabilistic literature \cite{CharlesM.Goldie,Teugels1975,Pitman1980,J.Chover1973}. This characterization holds for a particular definition of condensation given in Section \ref{sub sec: ips}, which was also used in \cite{ferrarietal07}. Our main result is that all spatially homogeneous processes with stationary product measures that exhibit condensation on finite lattices with a finite critical density are necessarily non-monotone.

Monotone (attractive) particle systems preserve the partial order on the state space in time, which enables the use of powerful coupling techniques to derive rigorous results on large scale dynamic properties such as hydrodynamic limits (see \cite{saada} and references therein). These techniques have also been used to study the dynamics of condensation in attractive zero-range processes with spatially inhomogeneous rates \cite{krugetal96,landim96,benjaminietal96,andjeletal00,ferrarisisko}, and more recently \cite{Bahadoran2014,Bahadoran2015}. As we discuss in Appendix \ref{sec:statmech}, non-monotonicity in homogeneous systems with finite critical density can be related, on a heuristic level, to convexity properties of the canonical entropy. 
For condensing systems with zero-range dynamics, it has been shown that this is related
to the presence of metastable states, resulting in the non-monotone behaviour of the canonical stationary current/diffusivity \cite{chlebounetal10}. This corresponds to a first order correction of a hydrodynamic limit leading to an ill-posed equation with negative diffusivity in the case of reversible dynamics. Heuristically, this is of course consistent with the concentration of mass in a small, vanishing volume fraction, but poses great technical difficulties to any rigorous proof of hydrodynamic limits for such particle systems. First results in this direction only hold for sub-critical systems under restrictive conditions \cite{Stamatakis2014}, and due to a lack of monotonicity there are no results for non-reversible dynamics.


Condensing monotone particle systems would, therefore, provide interesting examples of homogeneous systems for which coupling techniques could be used to derive stronger results on hydrodynamic limits. However, our result implies that this is not possible for condensing models with stationary product measures and a finite critical density on finite lattices. In the thermodynamic limit condensation has been defined through the equivalence of ensembles, which can be established in generality for a class of long-tailed distributions with a finite critical density \cite{stefan,Chleboun2013a}. This class has also been studied before \cite{Teugels1975,Pitman1980} and includes the class of sub-exponential distributions, for which our results apply also in the thermodynamic limit. A detailed discussion of their connections and the resulting differences between condensation on finite lattices and in the thermodynamic limit is given in Sections \ref{sec:subexp} and \ref{tlimit}. We remark that for systems where the dynamics is directly defined on infinite lattices there are no rigorous results or characterizations of condensation to our knowledge, and we do not discuss this case here.


For systems with infinite critical density condensation can still occur on finite lattices, and since non-monotonicity typically occurs above the critical density, such processes can also be monotone. 
When the tail of the stationary measure is a power law and decays faster than $n^{-3/2}$ with the occupation number $n$, we prove that the process is still non-monotone.
 In Section \ref{examples} we present preliminary results for tails that decay slower than $n^{-3/2}$, which strongly suggests that a monotone and condensing particle system exists (see \cite{Fajfrova2014} for further discussion).

The paper is organised as follows. In Section \ref{sec:notation} we introduce the background used to study condensation and monotonicity in particle systems, and state our main results. In Section \ref{sec:proof} we prove our main theorem by induction over the size of the lattice, showing that the family of canonical stationary measures is necessarily not monotonically ordered in the number of particles. 
In Section \ref{sec: char cond} we discuss the differences between condensation on fixed lattices and in the thermodynamic limit, and prove equivalence of condensation on finite lattices with the tail of the maximal invariant product measure being sub-exponential. 
In Section \ref{examples} we review examples of homogeneous processes that have been shown to exhibit condensation, and present some explicit computations for misanthrope processes and processes with power law tails.\\

\section{Notation and results}\label{sec:notation}

\subsection{Condensing stochastic particle systems}
\label{sub sec: ips}
We consider stochastic particle systems on fixed finite lattices $\Lambda =\{ 1,\ldots ,L\}$, which are continuous-time Markov chains on the countable state space $\Omega_{L} = \mathbb{N}^{\L}$. For a given configuration $\eta = (\eta_{x}:x \in \Lambda) \in \Omega_{L}$ the local occupation, $\h_x$ for $x \in \Lambda$, is a priori unbounded. The jump rates from configuration $\eta \in \Omega_{L}$ to $\xi \in \Omega_{L}$ are denoted $c(\eta, \xi)\geq 0$, and the dynamics of the process is defined by the generator
\begin{align}\label{gene}
\mathcal{L}f(\eta) = \sum_{\{\xi \in \Omega_{L} : \xi \neq \eta \} }{c(\eta,\xi)\big(f(\xi)-f(\eta)\big)} \ ,
\end{align}
for all continuous functions $f:\Omega_{L} \to \bbR$. 
We assume the process conserves the total number of particles
\begin{align*}
S_L (\eta ):=\sum_{x\in\Lambda} \eta_x \ ,
\end{align*}
and conditioned on $S_L =N$, the process is assumed to be irreducible, so that $S_L$ is the only conserved quantity. The process is therefore ergodic on the finite state space $\Omega_{L,N} = \{ \eta \in \Omega_{L} : S_L (\eta ) = N\}$ for all fixed $N\geq 0$. On $\Omega_{L,N}$ the process has a unique stationary distribution $\pi_{L,N}$, and the family $\big\{\pi_{L,N}:N\geq 0 \big\}$ is called the canonical ensemble. 

We focus on systems for which the stationary distributions are spatially homogeneous, \ie the marginal distributions $\pi_{L,N} [\eta_x \in .]$ are identical for all $x\in\Lambda$. This typically results from translation invariant dynamics on translation invariant lattices with periodic boundary conditions, but the actual details of the dynamics are not needed for our results. For these systems we define condensation in terms of the maximum occupation number
\begin{equation}
M_L (\eta ):= \max_{x\in\Lambda} \eta_x \ .
\label{maximum}
\end{equation}

\begin{definition}\label{condensation}
A stochastic particle system with canonical measures $\pi_{L,N}$ on $\Omega_{L,N}$ with $L\geq 2$ exhibits \textbf{condensation} if 
\begin{equation}
\lim_{K\to \infty}\lim_{N \to \infty}\pi_{L,N}[M_L \geq N-K] = 1 \ .
\label{eqn:cond1}
\end{equation}
\end{definition}

This condition implies in particular the existence of all limits involved. The interpretation of (\ref{eqn:cond1}) is that due to the sub-exponential tails of the measure, in the limit $N\to\infty$ all but finitely many particles concentrate in a single lattice site and that the distribution of particles outside the maximum is non-degenerate. As we will see in Proposition \ref{prop1} below, the latter is in fact given by the maximal invariant measure when the system exhibits product stationary measures.

There are of course other possible definitions of condensation which are less restrictive or more appropriate in other situations. For inhomogeneous systems or thermodynamic limits (with $N,L\to\infty$) the condensed phase can be localized in particular sites or have a more complicated spatial structure (see e.g. \cite{Waclaw2007} and for recent summaries \cite{Chleboun2013a,Godreche2012,Evans2014}). For our case of spatially homogeneous systems on finite lattices, $\Lambda$ with fixed $L$, the above is the most convenient definition and has been used in previously studied examples \cite{ferrarietal07}. A more detailed discussion is provided in Section \ref{sec: char cond}.\\


\subsection{Monotonicity and product measures}
\label{sec:mon}
We use the natural partial order on the state space $\Omega_{L}$ given by $\eta \leq \zeta$ if and only if $\eta_{x} \leq \zeta_{x}$ for all $x \in \Lambda$.  A function $f: \Omega_{L}\to \mathbb{R}$ is said to be increasing if $\eta \leq \zeta $ implies $ f(\eta) \leq f(\zeta)$. Two measures $\mu,\mu^{\prime}$ on $\Omega_{L}$ are stochastically ordered with $\mu \leq \mu^{\prime}$, if for all increasing functions $f: \Omega_{L}\to \mathbb{R}$ we have $\mu(f) \leq \mu^{\prime}(f)$, where $\mu(f)$ denotes the expectation of $f$ with respect to $\mu$.

A stochastic particle system on $\Omega_{L}$ with generator $\mathcal{L}$ and semi-group $(S(t)=e^{t\mathcal{L}} : t\geq 0)$ is called monotone (attractive) if it preserves stochastic order in time, i.e.
\begin{align*}
\mu \leq \mu^{\prime} \quad \implies \quad \mu S(t) \leq \mu^{\prime} S(t) \quad \textrm{for all } t \geq 0 \textrm{.}
\end{align*}

Coupling techniques for monotone processes are an important tool to derive rigorous results on the large scale dynamics of such systems such as hydrodynamic limits. There are sufficient conditions on the jump rates (\ref{gene}) to ensure monotonicity for a large class of processes (see e.g. \cite{saada} for more details), however for our results we only need a simple consequence for the stationary measures of the process.

\begin{lemma}\label{order}
If the stochastic particle system as defined in Section \ref{sub sec: ips} is monotone, then the canonical distributions $\pi_{L,N}$ are ordered in $N$, i.e.
\begin{equation}
\pi_{L,N} \leq \pi_{L,N+1} \quad\mbox{for all }N\geq 0\ .
\label{eq:order}
\end{equation}
\end{lemma}
The proof is completely standard but short, so we include it for completeness.

\begin{proof}
Fix a monotone process as defined in Section~\ref{sub sec: ips}. Consider two initial distributions $\mu$ and $\mu'$, concentrating on $\Omega_{L,N}$ and $\Omega_{L,N+1}$ respectively, given by
\begin{equation}
\nonumber
\mu[\eta] = \1(\eta_{1} = N) \quad \textrm{and}\quad \mu'[\xi] = \1(\xi_1 = N+1)\ ,
\end{equation}
for $\h \in \Omega_{L,N}$ and $\x \in \Omega_{L,N+1}$. Clearly $\mu \leq \mu'$, and so by monotonicity of the process this implies $\mu S(t) \leq \mu^{\prime} S(t)$ for all $t \geq 0$. Furthermore, by ergodicity we have
\begin{equation}
\nonumber
\pi_{L,N} =\lim_{t\to\infty} \mu S(t) \leq \lim_{t\to\infty} \mu' S(t)=\pi_{L,N+1} \ .
\end{equation}
\end{proof}

All rigorous results on condensing particle systems so far have been achieved for processes which exhibit stationary product measures, for which the measures $\pi_{L,N}$ take a simple factorized form. These can then be expressed in terms of un-normalized single-site weights $w(n)> 0$, $n\in\bbN$. Due to conservation of $S_L$ such processes exhibit a whole family of stationary homogeneous product measures
\begin{equation}
\nu^{L}_\phi [\eta ]=\prod_{x\in\Lambda} \nu_\phi [\eta_x ]\quad\mbox{with marginals}\quad \nu_\phi [\eta_x ]=\frac{w(\eta_x )}{z(\phi )}\,\phi^{\eta_x}\ .
\label{eq:pm}
\end{equation}
The measures are defined whenever the normalization
\begin{align}
z(\phi) := \sum_{n=0}^{\infty}{\phi^n w(n) }
\label{zphi}
\end{align}
is finite. This is the case for all fugacity parameters $\phi\in D_\phi$ where $D_\phi =[0,\phi_c )$ or $[0,\phi_c ]$, and
\begin{equation}\label{phic}
\phi_c :=\big(\limsup_{n\to\infty} \sqrt[n]{w(n)}\big)^{-1}
\end{equation}
is the radius of convergence of (\ref{zphi}). The family $\big\{ \nu_\phi :\phi\in D_\phi \big\}$ is also called the grand-canonical ensemble and $z(\phi )$ the (grand-canonical) partition function. Since the process is irreducible on $\Omega_{L,N}$ for all $N \in \bbN$ we have $w(n)>0$ for all $n \geq 0$. The canonical distribution can be written as
\begin{equation}
\pi_{L,N} [\eta ]=\nu^{L}_\phi [\eta |S_L =N]\quad\mbox{for all }\phi\in D_\phi \ ,
\end{equation}
which is independent of the choice of $\phi$. Equivalently
\begin{align}
\pi_{L,N}[\eta] = \frac{1}{Z_{L,N}} \prod_{x\in\Lambda} w(\eta_{x})\quad\mbox{where}\quad Z_{L,N} =\sum_{\eta\in \Omega_{L,N}} \prod_{x\in\Lambda} w(\eta_{x})
\label{eqn: prod measure}
\end{align}
is the (canonical) partition function. Note that throughout the paper we characterize all measures by their mass functions since we work only on a countable state space $\Omega_{L}$ and the measures $\pi_{L,N}$ concentrate on finite state spaces $\Omega_{L,N}$, which is the framework we rely on in this paper.\\

\subsection{Results}
\label{sec:resu}

Our results hold for systems with general stationary weights, $w(n)>0$ for each $n \in \bbN$, subject to the regularity assumption that
\begin{equation}
\lim_{n\to\infty} w(n-1)/w(n)\in (0,\infty ]
\label{regu}
\end{equation}
exists, which is then necessarily equal to $\phi_c$. 
If $\phi_{c}< \infty$ then weights that satisfy \eqref{regu} are sometimes called long-tailed \cite{foss2011introduction}, which is discussed in more detail in Section \ref{tlimit}.


For processes with such stationary product measures there is a simple equivalent characterization of condensation which we prove in Section \ref{proof equiv}.

\begin{proposition}\label{prop1}
Consider a stochastic particle system as defined in Section \ref{sub sec: ips} with stationary product measures as defined in Section \ref{sec:mon} satisfying regularity assumption \eqref{regu}. Then the process exhibits condensation according to Definition \ref{condensation} if and only if $\phi_c <\infty$, $D_\phi =[0,\phi_c ]$, and
\begin{equation}
\lim_{N\to\infty} \frac{\nu_{\phi_c}^{2} [\eta_1 +\eta_2 =N]}{\nu_{\phi_c} [\eta_1 =N]} =\lim_{N\to\infty} \frac{Z_{2,N}}{w(N)z(\phi_c )}\in (0,\infty )\quad\mbox{exists}\ .
\label{paff}
\end{equation}
In this case, the distribution of particles outside of the maximum converges weakly (equivalently in total variation) to the critical product measure $\nu_{\phi_c}^{L-1}$, \textit{i.e.} for fixed $n_1 , \ldots , n_{L-1}$ we have
\begin{align}
\label{weaklim}
\pi_{L,N}[\eta_{1} = n_{1},\ldots ,\eta_{L-1}=n_{L-1} | M_{L} = \eta_{L}] \to \prod_{i=1}^{L-1}\nu_{\phi_{c}}[\eta_{i} = n_i] \quad \textrm{as}\quad N \to \infty \ .
\end{align}
\end{proposition}
Note that for $\phi_c \in (0,\infty )$ we may rescale the exponential part of the weights to get $\phi_c =1$ and we can further multiply with a constant, so that in the following we can assume without loss of generality that
\begin{equation}
w(0)=1\quad\mbox{and}\quad \phi_c =\lim_{n\to\infty} w(n-1)/w(n) =1\ .
\label{wlog}
\end{equation}
The condition (\ref{paff}) can also be written as
\begin{equation}
\lim_{N\to\infty} \frac{Z_{2,N}}{w(N)}=\lim_{N\to\infty} \frac{(w*w)(N)}{w(N)}\in (0,\infty )\quad\mbox{exists}\ ,
\label{paff2}
\end{equation}
where $(w*w)(N)=\sum_{k=0}^N w(k)w(N-k)$ is the convolution product. This is a standard characterization to define the class of sub-exponential distributions (see e.g. \cite{Klppelberg1989, Baltrunas2004}). Sub-exponentiality implies that a large sum of two random variables is typically realized by one of the variables taking a large value (see Section \ref{sec:subexp} for more details), which is of course reminiscent for the concept of condensation. Implications and simpler necessary conditions on $w(n)$ which imply (\ref{paff2}) have been studied in detail, and we provide a short discussion in Section \ref{sec: char cond}. 
As a consequence of Proposition \ref{prop1}, condensation is only a property of the tail of $w$, therefore if a process with stationary product measures \eqref{eqn: prod measure} condenses in the sense of Definition \ref{condensation} for some $L\geq 2$, it condenses for all $L \geq 2$. 

Proposition \ref{prop1} provides a generalization of previous results on condensation on finite lattices \cite{ferrarietal07} and is used in the proof of our main result, which is the following.

\begin{theorem}\label{mainres}
Consider a spatially homogeneous stochastic particle system as defined in Section \ref{sub sec: ips} which exhibits condensation in the sense of Definition \ref{condensation}, has stationary product measures that satisfy \eqref{regu}, and has finite critical density
\begin{equation}
\rho_c :=\nu_{\phi_c} (\eta_1) =\frac{1}{z(\phi_c )}\sum_{n=0}^\infty nw(n)\phi_c^n <\infty\ ,
\label{fmean}
\end{equation}
then the process is necessarily not monotone.\\
The same is true if (\ref{fmean}) is replaced by the assumption that\footnote{For functions $g,h: \mathbb{N} \to \mathbb{R}$ we use the notation $g(n) \sim h(n)$ if $\frac{g(n)}{h(n)} \to  c \in (0,\infty)$ as $n\to\infty$.} $w(n)\sim n^{-b}$ 
 with $b\in (3/2,2]$.
\end{theorem}

\subsection{Discussion}\label{sec:discussion}

The class of distributions which fulfil \eqref{paff} (called sub-expontial), and therefore exhibit condensation on finite lattices, is large (see e.g. \cite[Table 3.7]{CharlesM.Goldie}), and includes in particular
\begin{itemize}
\item power law tails
$w(n) \sim n^{-b}$ where $b>1$, 
\item log-normal distribution
\begin{equation}
w(n) = \frac{1}{n}\exp \{ -(\log(n)-\mu)^2 / (2 \sigma^2) \}\ ,
\label{lognormal}
\end{equation}
where $\mu \in \mathbb{R}$ and $\sigma > 0$, which always has finite mean,
\item stretched exponential tails $w(n) \sim \exp \{ -Cn^{\gamma}\}$ for $0<\gamma<1$, $C>0$,
\item almost exponential tails $w(n) \sim \exp \big\{-\frac{n}{log(n)^{\beta}} \big\}$ for $\beta>0$.
\end{itemize}
For the last two examples all polynomial moments are finite. This covers all previously studied models on condensation in zero-range processes \cite{evans00,ferrarietal07,agl}. 
As we will discuss in Section \ref{tlimit} all these examples also exhibit condensation in the thermodynamic limit, whenever they have finite first moment. 
It can also be shown that the limit in (\ref{paff2}) is necessarily equal to $2z(\phi_c )$ and that in fact $\frac{Z_{L,N}}{w(N)}\to Lz(\phi_{c})^{L-1}$ for any fixed $L\geq 2$ (see  \cite{J.Chover1973} and Proposition \ref{2tol}). 


Since we consider a fixed lattice $\Lambda$, $\rho_c <\infty$ is not a necessary condition for condensation as opposed to systems in the thermodynamic limit. Even if the distribution of particles outside the maximum has infinite mean, condensation in the sense of Definition \ref{condensation} can occur. However, if $z(\phi_c )=\infty$ (e.g.\ for power law tails with $b\leq 1$), the distribution of particles outside the maximum cannot be normalized, condition \eqref{paff} fails, and there is no condensation in the sense of our definition.


\begin{figure}[t]
\centering
\begin{subfigure}{0.49\linewidth}
\caption{\label{back:powerlaw}}
\centering
\includegraphics[scale=1.,clip=true,trim= 2em 0 0em 0]{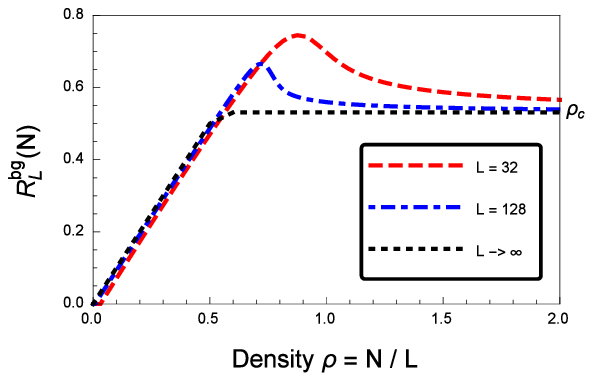}
\end{subfigure}
\centering
\begin{subfigure}{0.5\linewidth}
\centering
\caption{}
\label{back:lognormal}
\includegraphics[scale=1,clip=true,trim= 2em 0 0em 0]{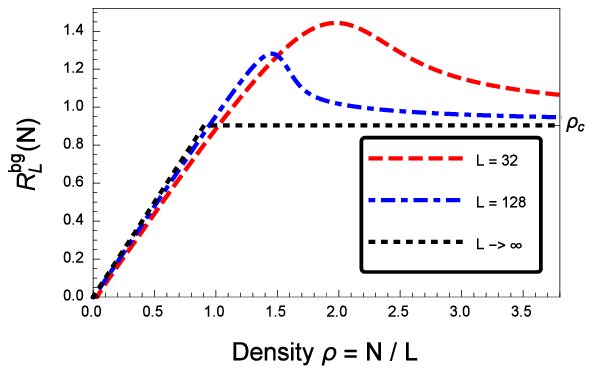}
\end{subfigure}

\begin{subfigure}{0.495\linewidth}
\caption{}
\centering
\label{back:almostexp}
\includegraphics[scale=1,clip=true,trim= 2em 0 0em 0]{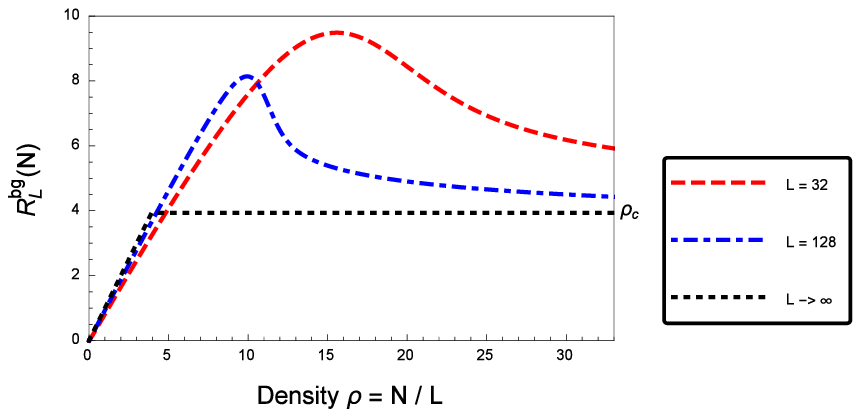}
\end{subfigure}
\caption{\label{backnumerics} Non-monotone behaviour of the expected background density $R_{L}^{bg}(N)$ \eqref{eq:backdensity} for lattice sizes $L=32$ and $L=128$; (A) (finite mean) power law tails with $b = 5$, (B) log-normal tails with $\mu = 0$ and $\sigma = 1/\sqrt{2}$, and (C) almost exponential tails with $\beta = 1$. The dotted black line shows the limit as $L,N \to \infty$ and $N/L\to\rho$, which is monotone and non-decreasing.
}
\end{figure}

We will prove non-monotonicity in the next section by showing that expectations for a particular monotone decreasing observable $f:\Omega_{L} \to\bbR$ under $\pi_{L,N}$ are not decreasing in $N$. The chosen function is related to (but not equal to) the number of particles outside the maximum (condensate), which has been shown previously to exhibit non-monotone behaviour for a class of condensing zero-range processes in the thermodynamic limit \cite{chlebounetal10,agl}. 
When the number of particles $N>\rho_c L$ just exceeds the critical value, typical configurations still appear homogeneous with a maximum occupation number of\footnote{For functions $f,g:\bbN \to \bbR$ we use the notation $g=o_{n}(f)$ if $\frac{g(n)}{f(n)} \to 0$ as $n \to \infty$.} $o_{N}(N)$. Only when the number of particles is increased further the system switches to a condensed state with a maximum that contains a non-zero fraction of all particles. We present numerical evidence of this non-monotone switching behaviour for the background density 
\begin{equation}
\label{eq:backdensity}
R_{L}^{bg}(N) := \frac{1}{L-1}\pi_{L,N}(N-M_{L})
\end{equation} 
in Figure \ref{backnumerics}. This is a finite size effect which disappears in the limit $L\to\infty$, and for specific models it has been shown to be related to the existence of super-critical homogeneous metastable states \cite{chlebounetal10,agl,Chleboun2015}. 
For large $L$, the switching to condensed states occurs abruptly over a relatively small range of values for $N$.
Since the $\pi_{L,N}$ are conditional product measures the correlations in the system are very weak, which causes metastable hysteresis effects and non-monotonicity of the canonical ensemble around the critical point.
Metastable hysteresis has been established in \cite{chlebounetal10,agl} for zero-range processes.
Our result implies that this behaviour is generic for all condensing systems with finite critical density. We also give a heuristic discussion of the connection to convexity properties of the entropy of the system in Appendix \ref{sec:statmech}.

%
There are several examples of homogeneous, condensing, monotone particle systems with finite critical density which have been studied on a heuristic level and which we summarize in Section \ref{examples}. Their stationary measures are not of product form and no explicit formulas are known, so these systems are therefore hard to analyse rigorously.
For systems with non-product stationary measures, upward fluctuations in the density which are homogeneously distributed may be suppressed strongly enough, so that the metastable states do not exist. Such models may then also be monotone, and examples are given in Section \ref{homnonpro}.

We excluded the case $\phi_{c} = 0$ in the presentation in Section \ref{sec:resu} for notational convenience, but it is easy to see that our results also hold in this case. With the convention $0^0=1$ we have $z(0)=w(0)=1$ and $\rho_c =0$, and then existence of the limit $Z_{2,N} /w(N)$ is equivalent to
\begin{equation*}
\pi_{2,N} [M_2 =N]=2\frac{w(N)w(0)}{Z_{2,N}} \to \frac{2w(0)}{2z(0)}=1\, \textrm{ as }N \to \infty ,
\end{equation*}
i.e.\ condensation of all $N$ particles on a single site. This can easily be extended to all $L\geq 2$ with Proposition \ref{2tol}. 
Considering only events with all $N$ particles on one site, or $N-1$ particles on one site  and $1$ particle elsewhere, we have convergence from above
\begin{equation*}
\frac{Z_{L,N}}{w(N)} - Lw(0)^{L-1}   >  L(L-1)w(0)^{L-2}w(1)\frac{w(N-1)}{w(N)} > 0 \ .
\end{equation*}
This implies the non-monotonicity of $\pi_{L,N}$ as discussed in Section \ref{sec:proof}.\\

\section{Proof of Theorem \ref{mainres}}\label{sec:proof}


We assume that the process exhibits condensation in the sense of Definition \ref{condensation} and has stationary product measures, so the canonical measures $\p_{L,N}$ are of the form \eqref{eqn: prod measure}. Furthermore, we assume the weights satisfy the regularity assumption, and without loss of generality $\phi_{c} =1$, see  \eqref{wlog}. We show that the family of canonical measures is not stochastically ordered in $N$, which implies non-monotonicity of the process by Lemma \ref{order}. To achieve this, we use the test function
\begin{equation}
f(\eta) := \1(  \eta_{1} = \eta_{2} = \ldots = \eta_{L-1} = 0  )\ ,
\label{eqn: dec observable}
\end{equation}
which indicates the event where all particles concentrate in the maximum at site $L$.



\begin{lemma}
\label{statement: monotonicity}
The function $f:\Omega_{L} \to\bbR$ defined in (\ref{eqn: dec observable}) is monotonically decreasing, which implies that
\begin{equation}
\frac{Z_{L,N}}{w(N)}\leq \frac{Z_{L,N+1}}{w(N+1)}\quad\mbox{for all }N\geq 0\ , 
\label{eq:zmono}
\end{equation}
whenever the family of canonical measures $\pi_{L,N}$ is stochastically ordered in $N$.
\end{lemma}

\begin{proof}

Fix configurations $\eta,\zeta \in \Omega_{L}$ such that $\eta \leq \zeta$. If $f(\eta) = 0$ then $\eta$ has at least one particle outside of site $L$, therefore so does $\z$ which implies $f(\zeta) = 0$. If $f(\eta) = 1$ then necessarily $f(\h)  \geq f(\z)$ since $f(\zeta) \in \{0,1 \}$. Therefore $f$ is a decreasing function. Using \eqref{eqn: prod measure} and the convention (\ref{wlog}), we find that the canonical expectation of the function \eqref{eqn: dec observable} is given by
\begin{align*}
\pi_{L,N} (f)=\frac{w(0)^{L-1} w(N)}{Z_{L,N}} =\frac{w(N)}{Z_{L,N}}\ .
\end{align*}
So if the canonical measures are monotone in $N$, monotonicity of $f$ implies (\ref{eq:zmono}).
\end{proof}

\begin{figure}[t]
\centering
\begin{subfigure}{0.49\linewidth}
\caption{ }
\label{test:powerlaw}
\centering
\includegraphics[scale=1.05,clip=true,trim= 1em 0 0em 0]{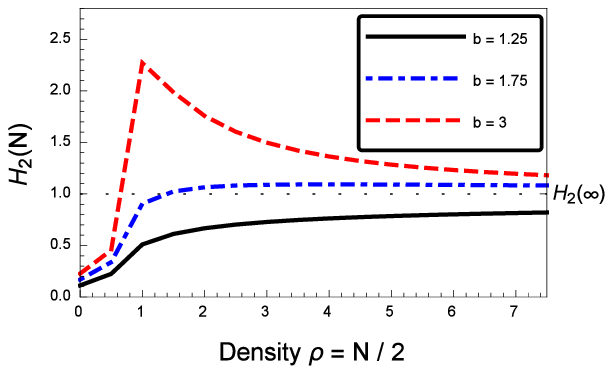}
\end{subfigure}
\centering
\begin{subfigure}{0.5\linewidth}
\centering
\caption{}
\label{test:lognormal}
\includegraphics[scale=1.05,clip=true,trim= 1em 0 0em 0]{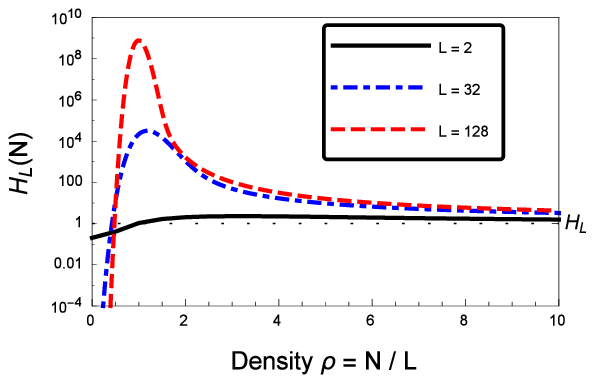}
\end{subfigure}
\caption{\label{testnumerics} Non-monotone behaviour of $H_{L}(N)$ \eqref{eq:scaledtest}, which is the expected value of the observable \eqref{eqn: dec observable} rescaled by its limit; (A) power law tails $w(n)\sim n^{-b}$ for $L=2$ with $b=3$, $1.75$ and $1.35$, where the latter is conjectured to be monotone (see Section \ref{examples}); (B) log-normal tails \eqref{lognormal} with $\mu = 0$ and $\sigma = 1/\sqrt{2}$ for $L=2$, $L=32$ and $L=128$.}
\end{figure}
By Proposition \ref{prop1} we know that for condensing systems the ratio $Z_{2,N}/w(N)$ converges. This implies the sequence $Z_{L,N}/w(N)$ in Lemma \ref{statement: monotonicity} converges, as summarised in the following proposition using our notation (see \cite[Theorem 1 and Lemma 5]{J.Chover1973}). 
\begin{proposition}\label{2tol}
Consider conditional product measures (\ref{eqn: prod measure}) with weights $w(n)>0$ for all $n \in \bbN$, which satisfy 
\begin{itemize}
\item $\frac{w(n-1)}{w(n)} \to \phi_{c}$ as $n \to \infty$, the regularity assumption (\ref{wlog}), \\
\vspace{-1em}
\item $z(\phi_{c}) < \infty$, \\ \vspace{-1em}
\item $\frac{Z_{2,N}}{w(N)} \to C$ as $N \to \infty$.
\end{itemize}
Then $C = 2 z(\phi_{c})$ and furthermore,
\begin{equation}
\label{eq:limitl}
\frac{Z_{L,N}}{w(N)} \to L z(\phi_{c})^{L-1} \textrm{ as }N \to \infty \textrm{ for all }L \geq 2 \ .
\end{equation}
\end{proposition}

Note that the limit in \eqref{eq:limitl} states that the probability of observing a large total number of particles under the critical product measure is asymptotically equivalent to the probability of observing a large number of particles on any one of the $L$ sites, precisely 
  \begin{align*}
    \lim_{N\to\infty} \frac{\nu_{\phi_c}^L\left[S_L(\h) = N\right]}{L\,\nu_{\phi_c}[\h_1 = N]} =1\,.
  \end{align*}
This is further equivalent to the canonical probability of the maximum containing the total mass 
converges to the critical probability that $L-1$ sites are empty, i.e.\ $\pi_{L,N} [ M_L = N] \to \nu_{\phi_c}^{L-1} [ \h \equiv 0 ]$.

To complete the proof we show that a subsequence of $Z_{L,N}/w(N)$ converges from above, which contradicts the assumption of monotonicity by Lemma \ref{statement: monotonicity}. We present a numerical illustration for the monotonicity properties of the function
\begin{equation} \label{eq:scaledtest} 
H_{L}(N) := \frac{1}{L z(1)^{L-1}}\frac{Z_{L,N}}{w(N)}  
\end{equation}
in Figure \ref{testnumerics}, which is normalized such that $H_L (N)\to 1$ as $N\to\infty$. 
The proof of the following lemma represents the most significant part of the proof of Theorem \ref{mainres} and is given in Section \ref{sec:finite-mean} for the case of finite mean and in Section \ref{sec:power-law} for the power law case.
\begin{lemma}\label{res2}
Under the conditions of Theorem \ref{mainres}, and assuming without loss of generality $\phi_c=1$, for each $L \geq 2$ there exists a $C > 0$ and a sequence $N_m \in\bbN$ with $N_m \to\infty$ as $m\to\infty$ such that
\begin{align*}
\min_{n\leq N_m}\bigg( \frac{Z_{L,n}}{w(N_{m})} - L z(1)^{L-1}  \bigg) \geq C/N_m \ ,
\end{align*}
\end{lemma}
Therefore, we know that there exists some $N^* \in \bbN$ such that $Z_{L,N^*}/w(N^*) > Z_{L,N^*+1}/w(N^*+1)$, which by Lemma \ref{statement: monotonicity} implies that the canonical measures are not stochastically ordered in $N$, and thus the process cannot be monotone by Lemma \ref{order}. This completes the proof of Theorem \ref{mainres}.




\subsection{Proof of Lemma \ref{res2}: The finite mean case}
\label{sec:finite-mean}

In order to prove that $\frac{Z_{2,N_m}}{w(N_m)}$ converges from above for some non-decreasing sequence $N_m$ we first specify a sequence on which we can bound the ratio $\frac{w(N_m -n)}{w(N_m)}$ below.


\begin{claim}
\label{statement: ratio upper bound}
For weights $\{ w(n) : n \in \mathbb{N} \}$ with finite and non-zero first moment, \textit{i.e.} $0<\rho_{c} < \infty$, there exists a sequence
$N_{m} \in \mathbb{N}$ with $N_{m} \to \infty$ as $m \to \infty$ such that for all $k \in \{0,\ldots , N_{m} - 1 \}$
\begin{equation}
\frac{w(N_{m}-k)}{w(N_{m})} \geq 1+\frac{k}{N_{m}} \textrm{.}
\end{equation}
\end{claim}
\begin{proof}

For each $m \in \mathbb{N}$, define $N_{m}$ as follows 
\begin{align*}
N_{m}=\max \{n \leq m : n \, w(n) = \min_{j \leq m} j\, w(j) \}\ .
\end{align*}
By definition $N_m$ is a non-decreasing sequence. Assume for contradiction that $N_m$ is bounded above, then for all $j \in \mathbb{N}$ we would have $j w(j) \geq j^{\star}w(j^{\star}) >0 $ for some $j^{\star} \in \mathbb{N}$, and therefore $\sum_n n w(n)\to \infty $ contradicting the assumption of finite first moment. For $k \in \{0,\ldots N_{m}-1 \}$ we have
\begin{align*}
(N_{m}-k)w(N_{m}-k) &\geq N_{m} w(N_{m}) \\
\mbox{and thus}\quad\frac{w(N_{m}-k)}{w(N_{m})} &\geq \frac{N_{m} }{(N_{m}-k)} \geq 1+\frac{k}{N_{m}} \textrm{.} 
\end{align*}
\end{proof}


\begin{claim}
\label{statement: 2 site limit above}
For weights $\{ w(n) : n \in \mathbb{N} \}$ with finite first moment,
there exists a subsequence $\{N_{\ell} : \ell \in \mathbb{N}\}$ of the sequence defined in Claim \ref{statement: ratio upper bound} such that $\frac{Z_{2,N_{\ell}-n}}{w(N_{\ell})}>2 z(1)$ for all $n \in \{0,\ldots ,N_{\ell} \}$ and $\ell$ sufficiently large.
\end{claim}
\begin{proof}
By neglecting at most a single term in the sum defining $Z_{2,N}$, the ratio $\frac{Z_{2,N}}{w(N)}$ can be bounded below as follows,
\begin{align}
\label{eq:Zbnd}
\frac{Z_{2,N}}{w(N)} = \sum_{n=0}^{N} w(n) \frac{w(N-n)}{w(N)} \geq 2 \sum_{n=0}^{\lfloor N/2 \rfloor -1}w(n) \frac{w(N-n)}{w(N)} \textrm{.}
\end{align}
We define
\begin{equation*}
K_m :=\max\Big\{ k^* \leq N_m :Z_{2,k^*}=\min_{0\leq k\leq N_m} Z_{2,k} \Big\}
\end{equation*}
to be the largest index where the ratio $\frac{Z_{2,k}}{w(N_m )}$ is minimized.
In particular 
\begin{align}
\label{eq:1}
 \frac{Z_{2,N_m - n}}{w(N_m )} \geq \frac{Z_{2,K_m}}{w(N_m )} \quad \textrm{for all} \quad  m \geq 0 \ \textrm{ and } \ n \in \{0,\ldots,N_m\}. 
\end{align}
By definition $K_m \leq N_m$, and so $r:= \limsup_{m\to\infty} K_m / N_m \leq 1$.
There exists a subsequence $(m_\ell :\ell \geq 0)$ such that $K_{m_\ell} / N_{m_\ell} \to r$, with a slight abuse of notation we denote the subsequences $N_{m_\ell}$ and $K_{m_\ell}$ simply by $N_\ell$ and $K_\ell$.
Suppose $r < 1$, by Claim \ref{statement: ratio upper bound} we have
\begin{align*}
\frac{Z_{2,K_\ell}}{w(N_\ell)} \geq \frac{Z_{2,K_\ell}}{w(K_\ell)}\left(2 - \frac{K_\ell}{N_\ell}\right) \to 2z(1) (2-r)> 2z(1) \,,
\end{align*}
which together with \eqref{eq:1} contradicts Proposition \ref{2tol},
therefore $K_{\ell} / N_{\ell} \to 1$ and $K_\ell / N_\ell \leq 1$ for all $\ell$. 

Applying Claim \ref{statement: ratio upper bound} we then have 
\begin{align*}
\frac{Z_{2,K_\ell}}{w(N_\ell)} \geq 2 \sum_{k=0}^{\lfloor K_\ell /2\rfloor -1 }{w(k)\frac{w(K_\ell -k)}{w(N_\ell )}} \geq 
\frac{2}{N_\ell}\sum_{k=0}^{\lfloor K_\ell /2\rfloor - 1}k\, w(k)+2\Big( 2-\frac{K_\ell}{N_\ell}\Big)\sum_{k=0}^{\lfloor K_\ell /2\rfloor - 1} w(k)\ .
\end{align*}
Subtracting $2z(1)$ we get
\begin{align}
\label{eq:2}
\frac{Z_{2,K_\ell}}{w(N_\ell)} &- 2z(1) \nonumber\\ & \geq\frac{2}{N_\ell}\sum_{k=0}^{\lfloor K_\ell /2\rfloor-1} k\,w(k) -2\sum_{k= \lfloor K_\ell /2 \rfloor}^\infty w(k)+2\Big( 1-\frac{K_\ell}{N_\ell}\Big)\sum_{k=0}^{\lfloor K_\ell /2\rfloor-1} w(k)\ .
\end{align}
Neglecting the final term in \eqref{eq:2} we have
\begin{align*}
  N_\ell\bigg(\frac{Z_{2,K_\ell}}{w(N_\ell)} &-2z(1)\bigg) > 2\sum_{k=0}^{\lfloor K_\ell /2\rfloor-1}k\, w(k) -2 N_\ell\sum_{k= \lfloor K_\ell /2 \rfloor}^\infty w(k)\\
&> 2\sum_{k=0}^{\lfloor K_\ell /2\rfloor-1}k\, w(k) -\frac{4 N_\ell}{K_\ell}\sum_{k= \lfloor K_\ell /2 \rfloor}^\infty k w(k)\to 2 \rho_c z(1) > 0\, ,
\end{align*}
as $\ell \to \infty$, where $\rho_c$ is the critical density defined in \eqref{fmean}. Together with \eqref{eq:1} this completes the proof of Claim \ref{statement: 2 site limit above}.
%
\end{proof}

To complete the proof of Lemma \ref{res2} we proceed by induction on the system size, $L$.
We make the following inductive hypothesis;  
\begin{enumerate}
\item[(H)]\label{H1} there exists a sequence $\{ N_{m}: m \in \mathbb{N}\}$ such that $\frac{Z_{L,N_{m} - n}}{w(N_{m})}>L z(1)^{L-1}$ for all $n \in \{0,\ldots , N_{m} \}$ and $m$ sufficiently large.
\end{enumerate}
The case $L=2$ is given by Claim \ref{statement: 2 site limit above}.
Analogously to the proof of Claim \ref{statement: 2 site limit above} we define
\begin{equation}
K_m :=\max\Big\{ k^* \leq N_m :Z_{L,k^*}=\min_{0\leq k\leq N_m} Z_{L,k} \Big\}.
\label{kmdef}
\end{equation}
By the same argument as in the proof of Claim \ref{statement: 2 site limit above}
there exists a  subsequence $(m_\ell :\ell \geq 0)$ such that $K_{m_\ell} / N_{m_\ell} \to 1$, again we denote the respective subsequences by $K_\ell$ and $N_\ell$.
For $\ell$ sufficiently large, we have
\begin{align*}
\frac{Z_{L+1,K_{\ell}}}{w(N_{\ell})} =& \sum_{k=0}^{K_{\ell}}{w(k)\frac{Z_{L,K_{\ell}-k}}{w(N_{\ell})}} 
=  \sum_{k=0}^{\lfloor K_{\ell}/2\rfloor}{w(k) \frac{Z_{L,K_{\ell}-k}}{w(N_{\ell})}} + \sum_{k=0}^{\lfloor K_{\ell}/2\rfloor-1}{Z_{L,k} \frac{w(K_{\ell}-k)}{w(N_{\ell})}} \\
&> L\, z(1)^{L-1} \sum_{k=0}^{\lfloor K_{\ell}/2\rfloor}{w(k)} + \sum_{k=0}^{\lfloor K_{\ell}/2\rfloor-1}{Z_{L,k} \bigg(1+\frac{N_{\ell}-K_{\ell}+k}{N_\ell}\bigg)} \textrm{,}
\end{align*}
where the final inequality follows from the inductive hypothesis (H) and Claim \ref{statement: ratio upper bound}.
Subtracting $(L+1) z(1)^{L}$ we get 
\begin{align}\label{arghh}
\frac{Z_{L+1,K_{\ell}}}{w(N_{\ell})} - (L + 1) z(1)^{L} >& -L\,z(1)^{L-1} \sum_{k=\lfloor K_{\ell} /2\rfloor +1}^{\infty}{w(k)}  \nonumber\\
&+ \frac{1}{N_{\ell}}\sum_{k=0}^{\lfloor {K_{\ell}}/{2}\rfloor-1}{k\,Z_{L,k}} + \bigg(1-\frac{K_{\ell}}{N_{\ell}} \bigg)\sum_{k=0}^{\lfloor{K_{\ell}}/{2}\rfloor-1}{Z_{L,k}} \nonumber\\
&+ \sum_{k=0}^{\lfloor{K_{\ell}}/{2}\rfloor-1}{Z_{L,k}}- z(1)^{L} \ .
\end{align}
Now, following the proof of Claim \ref{statement: 2 site limit above}, multiply (\ref{arghh}) by $N_{\ell}$ and neglect the second term on the second line.
Then the first term vanishes, since
\begin{align*}
0 \leq N_{\ell}\sum_{k=\lfloor K_{\ell}/2\rfloor +1}^{\infty}{w(k)} \leq \frac{2N_\ell}{K_\ell}\sum_{k=\lfloor K_{\ell}/2\rfloor +1}^{\infty}{k\,w(k)} \to 0 \quad\textrm{as } \ell\to \infty\ .
\end{align*}
In terms of the normalized grand-canonical measure $Z_{L,k}=z(1)^{L} \nu_1^{L} [S_{L} =k]$, so we have
\begin{equation}
\label{eq:3}
\sum_{k=0}^\infty k\, Z_{L,k} =z(1)^{L} \nu^{L}_{1} (S_{L} )=\rho_c L z(1)^{L}  \in (0,\infty )\ ,
\end{equation}
where $\rho_c$ is the critical density as defined in (\ref{fmean}). This implies that the first term in the second line of (\ref{arghh}), after multiplication with $N_\ell$, converges to a strictly positive constant. 
Finally, the third line in (\ref{arghh}) converges to zero after multiplying by $N_\ell$ since we have $\sum_{n=0}^{\infty}{Z_{L,n}} = z(1)^{L}$, which implies
\begin{align*}
0 \geq N_{\ell}\bigg(\sum_{k=0}^{\lfloor{K_{\ell}}/{2}\rfloor -1}{Z_{L,k}} - z(1)^{L} \bigg) = -N_\ell \sum_{k=\lfloor{K_{\ell}}/{2}\rfloor}^{\infty}{Z_{L,k}} \geq -\frac{N_\ell}{\lfloor{K_{\ell}}/{2}\rfloor} \sum_{k=\lfloor{K_{\ell}}/{2}\rfloor}^{\infty}k{Z_{L,k}}\to 0 
\end{align*}
as $\ell \to \infty$, by \eqref{eq:3}.
Using the definition of $K_\ell$ in (\ref{kmdef}), this implies that there exists a constant $C>0$ such that for all $\ell$ large enough
\begin{align*}
\min_{n\leq N_\ell}\bigg( \frac{Z_{L+1,n}}{w(N_{\ell})} - (L+1) z(1)^{L}  \bigg) \geq C/N_\ell \ ,
\end{align*}
so (H) holds for $L+1$, completing the induction. This concludes the proof of Lemma \ref{res2} for the case where the critical measure has finite mean.


\subsection{Proof of Lemma \ref{res2}: The infinite mean power law case}
\label{sec:power-law}

We consider stationary weights of the form $w(n) = n^{-b}h(n)$ with $w(0) = 1$, $h(n) \to c \in (0,\infty)$, and $b \in (1,2)$. We prove non-monotonicity of $Z_{L,N}/w(N)$ for $b \in (3/2,2)$ and $h(n) = 1$ for all $n \in \bbN$ via an exact computation. The case $b=2$ can be done completely analogously but involves different expressions with logarithms in the resulting limits, and is presented in Appendix \ref{bis2}. The proof remains valid for general converging $h(n)$ with only minor differences, which we explain in a remark at the end of this section. Convergence of $Z_{2,N}/w(N) \to 2z(1)$ from above or below for the exact power law depends on the parameter $b \in (1,2)$, as summarized in the next result.

\begin{lemma}
\label{lemma: infinite mean}
For stationary weights of the form $w(n) = n^{-b}$ and $w(0) = 1$ with $b \in (1,2)$
\begin{align}
\label{eqn: b less than}
N^{b-1}\bigg(\frac{Z_{2,N}}{w(N)} - 2z(1) \bigg) \to F_{2}(b)
\textrm{ as } N \to \infty \ ,
\end{align}
where
\begin{equation}
\nonumber
F_{2}(b) = 2\sum_{i=1}^{\infty}{\frac{1}{i!}\prod_{j=0}^{i-1}{\big(j+b\big)}\frac{2^{b-1-i}}{1-b+i} } - 2\frac{2^{b-1}}{b-1} \begin{cases} >0 &\textrm{ if } b\in (\frac{3}{2},2) \\  <0 &\textrm{ if } b\in (1,\frac{3}{2}) \end{cases} \ .
\end{equation}
For $L>2$ we have
\begin{align}
\label{llarger}
\lim_{N \to \infty}N^{b-1}\bigg(\frac{Z_{L,N}}{w(N)} - Lz(1)^{L-1} \bigg) = F_{L}(b) := z(1)F_{L-1}(b) + (L-1)z(1)^{L-2}F_{2}(b) \ ,
\end{align}
which has the same sign as $F_{2}(b)$.
\end{lemma}
This result implies that whenever $w(n) = n^{-b}$ for $n \geq 1$ and $w(0) = 1$ with $b \in (3/2,2)$ Lemma \ref{res2} holds with $C=F_2(b)$.
This completes the proof of Lemma \ref{res2} in the case $h(n)=1$. 

\begin{proof}[Proof of Lemma \ref{lemma: infinite mean}]
To prove this result we make use of the full Taylor series of $(1-x)^{-b}$ at $x=0$ and integral approximations to compute the asymptotic behaviour of summations. To simplify notation we assume that $N$ is even. For odd $N$ there is no term with multiplicity one and there exists an obvious modification. First note that $w(n)$ fulfils the regularity assumption \eqref{wlog} and $Z_{L,N}/w(N) \to L z(1)^{L-1}$ as $N\to \infty$ for all $L \geq 2$ \cite{ferrarietal07}, so by Proposition \ref{prop1} a process with stationary measures $\pi_{L,N}$ will exhibit condensation. For $L=2$ we subtract $2z(1)$ from $Z_{2,N}/w(N)$ to get
\begin{align}
\nonumber
\frac{Z_{2,N}}{w(N)} - 2z(1) &= 2\sum_{n=0}^{N/2}{w(n) \frac{w(N-n)}{w(N)}} - 2 \sum_{n=0}^{\infty}{w(n)} - \frac{w(N/2)w(N/2)}{w(N)}  \\
\label{eq:first}
&=2\sum_{n=0}^{N/2}{n^{-b} \bigg(1-\frac{n}{N} \bigg)^{-b}} - 2 \sum_{n=0}^{\infty}{n^{-b}} - 2^{2b}N^{-b}\ .
\end{align}
Substituting the Taylor expansion of $(1-x)^{-b}$ we find
\begin{align}
\nonumber
\frac{Z_{2,N}}{w(N)} - 2z(1) &= 2\sum_{n=0}^{N/2} n^{-b}\sum_{i=0}^{\infty} \frac{1}{i!} \bigg(\frac{n}{N}\bigg)^i \prod_{j=0}^{i-1} (j+b) - 2 \sum_{n=1}^{\infty}{n^{-b}}- 2^{2b}N^{-b}  \\
\label{powerlaw difference}
&=2\sum_{i=1}^{\infty}{\frac{1}{i!}\prod_{j=0}^{i-1}{\big(j+b\big)} \frac{1}{N^{i}}\sum_{n=1}^{N/2}{n^{-b+i}}} - 2 \sum_{n=N/2+1}^{\infty}{n^{-b}} - 2^{2b}N^{-b}\ .
\end{align}
In the last line the $i=0$ term was combined with the second term, and we adopt the usual convention that empty products are equal to one. Both summations in $n$ are over continuous and monotone functions $g:\bbR\to (0,\infty )$, therefore we can use the usual integral approximation for decreasing (increasing) functions 
\begin{equation}
\label{eq:intbounds}
\int_c^{d+1} g(x)\, dx\leq (\geq) \sum_{n=c}^d g(n)\leq (\geq) g(c)+\int_c^d g(x)\, dx
\end{equation} 
for all $c\in\N$ and $d\in\N\cup\{\infty\}$. Multiplying by $N^{b-1}$ we find the limit as $N \to \infty$ of (\ref{powerlaw difference}) to be
\begin{align}
F_{2}(b) = 2\sum_{i=1}^{\infty}{\frac{1}{i!}\prod_{j=0}^{i-1}{\big(j+b\big)}\frac{2^{b-1-i}}{1-b+i} } - 2\frac{2^{b-1}}{b-1}\ .
\end{align}
It is shown in Appendix \ref{appendixSign} that this is positive (and finite) in the region $b \in (3/2,2)$ and negative (and finite) in the region $b \in (1,3/2)$, completing the proof of Lemma \ref{lemma: infinite mean} for $L = 2$. The result holds for general system size, $L \geq 2$, and is proved by induction. The inductive hypothesis states
\begin{equation}
\label{eq:induction hyp}
\lim_{N \to \infty}N^{b-1}\bigg(\frac{Z_{L,N}}{w(N)} - L z(1)^{L-1}\bigg) = F_{L}(b) = z(1)F_{L-1}(b) + (L-1)z(1)^{L-2}F_{2}(b) \ .
\end{equation}

Similar to the case $L=2$ we write
\begin{align}
\nonumber
N^{b-1}\bigg(\frac{Z_{L+1,N}}{w(N)} &- (L+1)z(1)^{L} \bigg)\\
= \underbrace{N^{b-1}\bigg(\sum_{n=0}^{N/2}{Z_{L,n}\frac{w(N-n)}{w(N)}}-z(1)^{L}\bigg)}_{\Xi_{L,N}} +& \underbrace{N^{b-1}\bigg(\sum_{n=0}^{N/2-1}{w(n)\frac{Z_{L,N-n}}{w(N)}}-Lz(1)^{L} \bigg)}_{\Theta_{L,N}}\ .
\label{chuck}
\end{align}
We first establish the limit of the function $\Theta_{L,N}$ in equation \eqref{chuck}. The inductive hypothesis \eqref{eq:induction hyp} can be written as
\begin{equation}
\label{eq:induction2}
\frac{Z_{L,n}}{w(n)} = \frac{F_{L}(b)+o_{n}(1)}{n^{b-1}} + Lz(1)^{L-1} \ ,
\end{equation} 
which implies $\Theta_{L,N}$ can be written as
\begin{align*}
\Theta_{L,N} &= N^{b-1}\left(\sum_{n=0}^{N/2-1}w(n) \frac{w(N-n)}{w(N)}\frac{Z_{L,N-n}}{w(N-n)} - L z(1)^{L}\right) \\
&= N^{b-1}\left(\sum_{n=0}^{N/2-1}w(n) \frac{w(N-n)}{w(N)}\left[ \frac{F_{L}(b)+o_{N}(1)}{(N-n)^{b-1}} + Lz(1)^{L-1}\right] - L z(1)^{L} \right) \ .
\end{align*}
Rearranging terms and noting that $\frac{w(N-n)}{w(N)}\frac{N^{b-1}}{(N-n)^{b-1}}=\left(\frac{N-n}{N}\right)^{1-2b}$ we then have
\begin{align}
\nonumber
\Theta_{L,N}=&(F_{L}(b) + o_{N}(1))\sum_{n=0}^{N/2-1}{w(n)\bigg(\frac{N-n}{N}\bigg)^{1-2b}} \\
\nonumber
 &+Lz(1)^{L-1}N^{b-1}\bigg( \sum_{n=0}^{N/2-1}{w(n)\frac{w(N-n)}{w(N)}} - z(1)\bigg) \ .
\end{align}
After Taylor expanding $(1-x)^{1-2b}$ appearing in the first line above, it is easy to see that the limit of the first line is given by $F_{L}(b) z(1)$ as $N\to \infty$. Using the $L=2$ result to calculate the limit of the second line we find
\begin{equation}
\label{eq:2ndbra}
\Theta_{L,N}\to F_{L}(b) z(1) + \frac{L z(1)^{L-1}F_{2}(b)}{2} \textrm{ as }N \to \infty \ .
\end{equation}
To identify the limit of $\Xi_{L,N}$ in \eqref{chuck}, we again make use of the Taylor expansion of $(1-x)^{-b}$ similarly to the two site case and we write
\begin{align}
\nonumber
\Xi_{L,N}=N^{b-1}\bigg( \sum_{n=0}^{N/2}{Z_{L,n} \sum_{i=0}^{\infty}{\frac{1}{i!}\prod_{j=0}^{i-1}\big(j+b\big) \bigg(\frac{n}{N}\bigg)^{i}}}-z(1)^{L} \bigg)\ .
\end{align}
Changing the order of summations, separating the $i = 0$ term 
and using $\sum_{n=0}^{\infty}{Z_{L,n}} = z(1)^{L}$ we have
\begin{align}
\label{eq:xi}
\Xi_{L,N}=N^{b-1}\bigg(\sum_{i=1}^{\infty}{\frac{1}{i!}\prod_{j=0}^{i-1}\big(j+b\big)\frac{1}{N^{i}} \sum_{n=1}^{N/2}{n^{i}Z_{L,n}}} - \sum_{n=N/2+1}^{\infty}{Z_{L,n}} \bigg)\ .
\end{align}
For all $i \geq 1$ and $b \in (1,2)$ we have $N^{b-1-i} \to 0$ as $N \to \infty$, which implies that for any fixed $N_{1} \in \bbN$ we have $N^{b-1-i}\sum_{n=1}^{N_{1}-1}{n^i Z_{L,n}} \to 0$. Therefore, the following limits are equal
\begin{align*}
\nonumber
\lim_{N \to \infty}\Xi_{L,N}= \lim_{N \to \infty}N^{b-1}\bigg(\sum_{i=1}^{\infty}{\frac{1}{i!}\prod_{j=0}^{i-1}\big(j+b\big)\frac{1}{N^{i}} \sum_{n=N_{1}}^{N/2}{n^{i}Z_{L,n}}} - \sum_{n=N/2+1}^{\infty}{Z_{L,n}} \bigg)\ .
\end{align*}
 Using the inductive hypothesis \eqref{eq:induction2} we have $\lim_{N \to \infty}\Xi_{L,N}$ is given by
\begin{align}
\nonumber
\lim_{N \to \infty}N^{b-1}(F_{L}(b)+o_{N}(1))&\bigg(\sum_{i=1}^{\infty}{\frac{1}{i!}\prod_{j=0}^{i-1}\big(j+b\big)\frac{1}{N^{i}} \sum_{n=N_{1}}^{N/2}{n^{i}\frac{w(n)}{n^{b-1}}}} - \sum_{n=N/2+1}^{\infty}{\frac{w(n)}{n^{b-1}}} \bigg)\\
+ \lim_{N \to \infty}N^{b-1}Lz(1)^{L-1}&\bigg(\sum_{i=1}^{\infty}{\frac{1}{i!}\prod_{j=0}^{i-1}\big(j+b\big)\frac{1}{N^{i}} \sum_{n=N_{1}}^{N/2}{n^{i}w(n)}} - \sum_{n=N/2+1}^{\infty}{w(n)} \bigg)\ .
\label{eq:xilim}
\end{align} 
Now applying the $L=2$ result it is possible to show that
\begin{equation}
\label{eq:1stbra}
\Xi_{L,N} \to \frac{L z(1)^{L-1}F_{2}(b)}{2} \ ,
\end{equation}
where the limit of the first line of \eqref{eq:xilim} was $0$ by the additional factor $1/n^{b-1}$ appearing in the summations.
Combining \eqref{eq:2ndbra} and \eqref{eq:1stbra} we have
\begin{equation}
\nonumber
N^{b-1}\bigg(\frac{Z_{L+1,N}}{w(N)} - (L+1)z(1)^{L} \bigg) \to z(1)F_{L}(b)+L z(1)^{L-1}F_{2}(b) \textrm{ as }N \to \infty \ ,
\end{equation}
concluding the induction so the result holds for all $L \geq 2$. From the recursion \eqref{llarger} it is obvious that $F_{L}(b)$ will have the same sign as $F_{2}(b)$, completing the proof of Lemma \ref{lemma: infinite mean}.
\end{proof}

A slightly modified version of Lemma \ref{lemma: infinite mean} also holds if the stationary weights are of the form $w(n) = n^{-b} h(n)$ where  $\lim_{n \to \infty}{h(n)} = c \in (0, \infty)$. The limit in \eqref{eqn: b less than} only depends on the tail behaviour of the weights and is now given by $c F_{2}(b)$. Briefly, this can be seen as follows, \eqref{eq:first} becomes
\begin{equation}
\nonumber
2\sum_{n=0}^{N/2}{n^{-b}h(n)\frac{h(N-n)}{h(N)} \bigg(1-\frac{n}{N} \bigg)^{-b}} - 2 \sum_{n=0}^{\infty}{n^{-b}h(n)} + 2^{2b}N^{-b} \frac{h(N/2)h(N/2)}{h(N)}\ .
\end{equation}
Taylor expanding $(1-x)^{-b}$ and rearranging to find terms of the form $N^{1-b-i}\sum_{n=1}^{N/2}{h(n)n^{-b+i}}$ and using the same argument to calculate the limit of $\Xi_{L,N}$ we have 
\begin{align}
\nonumber
\lim_{N\to \infty}N^{b-1-i}\sum_{n=1}^{N/2}{h(n) n^{-b+i}}
=\lim_{N \to \infty}{ N^{b-1-i}\sum_{n=N_{1}}^{N/2}{cn^{-b+i}}}  < \infty
\end{align}
for all $i \geq 1$ and any $N_{1} \in \bbN$, 
and the result follows. Similar modifications are required in the inductive step and the new limit in \eqref{llarger} is given by $c^{L-1}F_{L}(b)$ for all $L \geq 2$. This does not change the sign of the limit in \eqref{llarger} and therefore Lemma \ref{res2} still holds.\\

\section{Characterization of condensation}
\label{sec: char cond}

Condensation arises in spatially homogeneous systems with stationary product measures due to the sub-exponential tail of the stationary weights $w$, which has been studied extensively in previous work. 
In this section we review relevant results on heavy-tailed distributions and discuss the links between condensation on finite fixed lattices and in the thermodynamic limit before we give the proof of Proposition \ref{prop1} in Section \ref{proof equiv}.


\subsection{Sub-exponential distributions}
\label{sec:subexp}
Sub-exponential distributions are a special class of heavy-tailed distributions, the following characterization was introduced in \cite{Chistyakov1964} with applications to branching random walks, and has been studied systematically in later work (see e.g.\ \cite{J.Chover1973,Teugels1975,Pitman1980,Klppelberg1989}), for a review see for example \cite{CharlesM.Goldie} or \cite{Baltrunas2004}.

A non-negative random variable $X$ with distribution function $F(x)= \bbP [X\leq x]$ is called heavy-tailed if $F(0+)=0$, $F(x)<1$ for all $x>0$, and
\begin{equation}
e^{\l x}(1-F(x)) \to \infty\,\textrm{ as }x\to \infty \,\textrm{ for all }\l>0 \ .
\end{equation}
It is called sub-exponential if $F(0+)=0$, $F(x)<1$ for all $x>0$, and
\begin{equation}
\frac{1-F^{\star 2}(x)}{1-F(x)}\to 2\quad\mbox{as }x\to\infty\ .
\label{eq:subex}
\end{equation}
Here $F^{\star 2} (x)=\bbP [X_1 +X_2 \leq x]$ denotes the convolution product, the distribution function of the sum of two independent copies $X_1$ and $X_2$. It has been shown \cite{Chistyakov1964,Embrechts1980} that \eqref{eq:subex} is equivalent to either of the following conditions,
\begin{align}
\label{sube1}&\lim_{x \to \infty} \frac{1-F^{\star L}(x)}{1-F(x)} = L \textrm{\quad for all \quad} L \geq 2 \ ,\quad\mbox{or}&\\
\label{sube2}&\lim_{x \to \infty} \frac{\mathbb{P}\big[\sum_{i=1}^{L}{X_{i}}>x\big]}{\mathbb{P}\big[\max\{X_{i}: i \in \{1, \ldots, L \} \}>x\big]} = 1 \textrm{\quad for all \quad} L \geq 2  \ . & 
\end{align}
The second characterization shows that a large sum of independent sub-exponential random variables $X_i$ is typically realized by one of them taking a large value, which is of course reminiscent of the condensation phenomenon. 
It was further shown in \cite{Chistyakov1964,CharlesM.Goldie} that sub-exponential distributions also have the following properties, 
\begin{align}
\lim_{x \to \infty}{\frac{1-F(x-y)}{1-F(x)}} =1 &\quad \forall y \in \mathbb{R}, \\
\int_{0}^{\infty}{e^{\epsilon x} dF(x)} = \infty &\quad \forall \epsilon >0 \quad \textrm{(no exponential moments),}\\
F(x)e^{\epsilon x} \to \infty &\quad \forall \epsilon > 0 \quad \textrm{(slower than exponential decay)}\ .
\end{align}
Most results in the literature are formulated in terms of distribution functions and tails and apply to discrete as well as continuous random variables. \cite{J.Chover1973} provides a valuable connection to discrete random variables in terms of their mass functions $w(n)$, $n\in\N$.\\
Assume the following properties for a sequence $\{w(n)>0 : n\in \mathbb{N} \}$,
\begin{itemize}
\item[(a)] $\frac{w(n-1)}{w(n)} \to 1$ as $n \to \infty$,
\item[(b)] $z(1):=\sum_{n=0}^{\infty}{w(n)} \in (0,\infty)$\quad (normalizability),
\item[(c)] $\lim_{N \to \infty}\frac{(w*w)(N)}{w(N)} = C \in (0,\infty)$ exists.
\end{itemize}
Then \cite[Theorem 1]{J.Chover1973} asserts that $C = 2z(1)$ and $w(n)/z(1)$ is the mass function of a discrete, sub-exponential distribution. The implication
$$
\frac{(w^{\star L})(N)}{w(N)}\to Lz(1)^{L-1}\quad\mbox{as }N\to\infty\mbox{ for }L>2
$$
is given in \cite[Lemma 5]{J.Chover1973}. Sufficient (but not necessary) conditions for assumption (c) to hold are given in \cite[Remark 1]{J.Chover1973}.\\
Provided $z(1)< \infty$, then (c) holds if either of the following conditions are met:
\begin{itemize}
\item[(i)] $\sup_{1 \leq k \leq n/2}{\frac{w(n-k)}{w(n)}} \leq K$ 
\end{itemize}
for some constant $K>0$, or
\begin{itemize}
\item[(ii)] $w(n) = e^{-n \psi(n)}$
\end{itemize}
where $\psi(x)$ is a smooth function on $\bbR$ with $\psi(x) \searrow 0$ and $x^2 |\psi^{\prime}(x)| \nearrow \infty$ as $x\to\infty$, and $\int_{0}^{\infty}{dx \, e^{-\frac{1}{2}x^2 |\psi^{\prime}(x)|} }< \infty$.\\
Case (i) includes distributions with power law tails, $w(n) \sim n^{-b}$ with $b>1$. The stretched exponential with $\psi(x) = x^{\gamma - 1}$, $\gamma \in (0,1)$, and the almost exponential with $\psi(x) = (\log(x))^{-\beta}$, $\beta >0$, are covered by case (ii). The class of sub-exponential distributions includes many more known examples than the list given in Section \ref{sec:discussion} (see e.g. \cite[Table 3.7]{CharlesM.Goldie}).
%
Analogous to the characterisation of sub-exponential distributions, given by \eqref{sube2}, for discrete distributions the existence of the limit $(w*w)(N)/w(N)$ is equivalent to the existence of the following condition
\begin{equation}
\frac{\mathbb{P}[X_{1}+X_{2} = N]}{\mathbb{P}\big[\max\{X_{1},X_{2}\} = N \big]} \to 1 \textrm{ as }N \to \infty \ . 
\end{equation}
This holds, since we have the following equality of ratios
\begin{equation}
\nonumber
\frac{\mathbb{P}[X_{1}+X_{2} = N]}{\mathbb{P}\big[\max\{X_{1},X_{2}\} = N\big]} = \frac{Z_{2,N}}{2w(N)\sum_{n=0}^{N}{w(n)}} = \frac{(w*w)(N)}{2w(N)\sum_{n=0}^{N}{w(n)}} \ .
\end{equation}
 
Specific properties of power law tails $w(n)$ are used in \cite{ferrarietal07} to show condensation for finite systems in the sense of Definition \ref{condensation}. In Proposition \ref{prop1}, proved in Section \ref{proof equiv}, we extend this result to stationary product measures with general sub-exponential tails. In this context, condensation is basically characterized by the property \eqref{sube2} which assures emergence of a large maximum when the sum of independent variables is conditioned on a large sum. As summarized in the introduction, condensation in stochastic particle systems has mostly been studied in the thermodynamic limit with particle density $\rho\geq 0$, where $L,N\to\infty$ such that $N/L\to\rho$. In that case conditions on the sum of $L$ independent random variables become large deviation events, which have been studied in detail in \cite{Denisov2008,Armendariz2011}.\\

\subsection{Connection with the thermodynamic limit}
\label{tlimit}

In the thermodynamic limit, a definition of condensation is more delicate and the approach presented in \cite{stefan,Chleboun2013a} follows the classical paradigm for phase transitions in statistical mechanics via the equivalence of ensembles (see e.g.\ \cite{georgii1979canonical} for more details). A system with stationary product measures \eqref{eq:pm} exhibits condensation if the critical density \eqref{fmean} is finite, \textit{i.e.} $\rho_{c} < \infty$ and the canonical measures $\pi_{L,N}$ are equivalent to the critical product measure $\nu_{\phi_{c}}$ in the limit $L,N\to\infty$ such that $N/L\to\rho$ for all super-critical densities $\rho\geq\rho_c$. The interpretation is again that the bulk of the system (any finite set of sites) is distributed as the critical product measure in the limit. It has been shown in \cite{stefan} (see also \cite{Chleboun2013a} for a more complete presentation) that the regularity condition \eqref{regu} and $\rho_c <\infty$ imply the equivalence of ensembles, which has therefore been used as a definition of condensation in \cite[Definition 2.1]{Chleboun2013a}. 
Therefore, any process that condenses for fixed $\L$ with $\rho_{c}< \infty$ in the sense of Definition \ref{condensation} also condense in the thermodynamic limit. 
This includes all previously studied examples \cite{evans00,ferrarietal07,agl}, however there exists distributions that satisfy \eqref{regu} with $\rho_{c}< \infty$ but do not satisfy the conditions of Proposition \ref{prop1} and do not condense for fixed $\L$.
This is illustrated by an example given below.
It is also discussed in \cite[Section 3.2]{Chleboun2013a} that assumption \eqref{regu} is not necessary to show equivalence, but weaker conditions are of a special, less general nature and are not discussed here. Note also that equivalence of ensembles does not imply that the condensate concentrates on a single lattice site, the latter has been shown so far only for stretched exponential and power-law tails with $\rho_c <\infty$ in \cite{armendarizetal09,agl}. 
Both definitions involve only the sequence $\pi_{L,N}$ of canonical measures and not on the dynamics of the underlying process. Since the canonical measures \eqref{eqn: prod measure} are fully characterised by the weights $w(n)$ condensation can be viewed as a feature of the tails of the weights $w(n)$.

The condensation phenomena can also be studied for continuous random variables on the local state space $[0,\infty)$, see for example \cite{Armendariz2011}. 
The following continuous example, taken from \cite{Pitman1980} is shown to satisfy \eqref{regu} but is not sub-exponential. We show that the distribution has a finite mean and therefore exhibits condensation in the thermodynamic limit but not on a finite lattice in the sense of Definition \ref{condensation}. 
For a real-valued random variable $X$ with distribution function $F(x)=\mathbb{P}[X \leq x]$, assume $F^{\prime}(x) =  g^{\prime}(x) e^{-g(x)}$. 
Let $(x_{n})_{n \in \bbN}$ to be an increasing sequence with $x_{0}=0$ and $g(x)$ be a continuous and piecewise linear function such that $g(0)=0$ and $g^{\prime}(x) = 1/n$ for $x \in (x_{n-1},x_{n})$. The sequence $(x_{n})_{n \in \bbN}$ is defined iteratively as follows

\begin{align*}	
x_{n}-x_{n-1} &=  2 n e^{g(x_{n-1})} \\
\nonumber
g(x_{n})-g(x_{n-1}) &= 2 e^{g(x_{n-1})} \ ,
\end{align*}
and $g(x) - g(x_{n-1}) = \frac{x-x_{n-1}}{n} $ for $x \in [x_{n-1},x_{n})$. The mean is finite since
\begin{align*}
\int_{0}^{\infty}{x F^{\prime}(x)dx} = \sum_{n=1}^{\infty}{\frac{1}{n}\int_{x_{n-1}}^{x_{n}} x e^{-g(x)}dx}= \sum_{n=0}^{\infty} e^{-g(x_{n})} < \infty \ ,
\end{align*}
where the final step uses the relation $g(x_{n})-g(x_{n-1}) = 2e^{g(x_{n-1})} \geq 2 (1+g(x_{n-1}))$ to bound the series from above.

For all distributions satisfying \eqref{regu} which are not sub-exponential $Z_{L,N}/w(N)$ does not have a limit in $(0,\infty)$ as $N \to \infty$ and with Proposition \ref{prop1} there is no condensation on finite lattices according to Definition \ref{condensation}. 
For a discretized version of the example given above with weights $w(k) = \exp \{-g(k)\}$ we have $Z_{2,N}/w(N) \to \infty$ for $N \sim x_{n}$ as $n \to \infty$ \cite{Pitman1980}. 
For this example, following the proof of Proposition \ref{prop1} this implies that $\pi_{2,N}[\eta_{1} \land \eta_{2}\leq K]\to 0$ for $N \sim x_{n}$ as $N \to \infty$ and all $K \geq 0$. 
Therefore, the $L=2$ bulk occupation number $\eta_{1} \land \eta_{2}$ diverges in distribution as $N \to \infty$ by receiving a diverging excess mass from the condensate due to the light tail of $w(n)$.
It can be shown that these weights satisfy \eqref{regu} and therefore exhibit condensation in the thermodynamic limit where the excess mass can be distributed on a diverging number of sites.

For a process that exhibits condensation in the thermodynamic limit with a sub-exponential critical product measure, Proposition \ref{prop1} implies that condensation occurs also on finite lattices with $\rho_{c}<\infty$. Theorem \ref{mainres} then implies that this process is necessarily non-monotone for all fixed system sizes $L$. However, monotonicity for condensing processes with long-tailed but not sub-exponential stationary measures remains open.

\subsection{Proof of Proposition \ref{prop1}}
\label{proof equiv}
Let us first assume that the process exhibits condensation according to Definition \ref{condensation} and has canonical distributions of the form (\ref{eqn: prod measure}) where the weights fulfil (\ref{regu}), i.e. $w(n-1)/w(n)\to\phi_c \in (0,\infty ]$ as $n\to\infty$. 
In this part of the proof we establish that;
\begin{enumerate}
\item $\phi_{c}< \infty$,
\item $\frac{Z_{L,N}}{w(N)}$ has a limit as $N \to \infty$,
\item $z(\phi_{c})< \infty$, which also implies $\frac{Z_{L,N}}{w(N)} \to L z(\phi_{c})^{L-1}$ as $N \to \infty$, and
\item convergence of $\frac{Z_{L,N}}{w(N)} \to L z(\phi_{c})^{L-1}$ for some $L \geq 2$ implies convergence for $L=2$ and therefore \eqref{paff} holds.
\end{enumerate}

Step (1), show $\phi_c <\infty$. 
Assume first that $w(n-1)/w(n) \to \infty$ as $n \to \infty$. 
For all $K \in \bbN$ and $N>K$ we have
\begin{align*}
\nonumber
\pi_{L,N}[M_{L}\geq N-K] &= \frac{L}{Z_{L,N}}\sum_{n=0}^{K}Z_{L-1,n}w(N-n) \\
&\leq L\frac{K+1}{Z_{L,N}}\max_{0\leq n \leq K}\left(Z_{L-1,n}\right)\max_{0\leq n \leq K} \left(w(N-n)\right) \ .
\end{align*}
Let $n^{\star} = \textrm{arg}\!\max_{0\leq n\leq K} (w(N-n)) \leq K$. The partition function $Z_{L,N}$ is trivially bounded below by the event that site 1 takes $N-n^{\star}-1$ particles and the second site takes the remaining $n^{\star}+1$ particles, \ie
\begin{equation}
\nonumber
Z_{L,N}\geq w(0)^{L-2}w(n^{\star}+1)w(N-n^{\star}-1) \ .
\end{equation}
Therefore
\begin{align*}
\pi_{L,N}[M_{L}\geq N-K] \leq \frac{L}{w(0)^{L-2}}\frac{K+1}{w(n^{\star}+1)} \frac{w(N-n^{\star})}{w(N-n^{\star}-1)}\max_{0\leq n \leq K}\left(Z_{L-1,n}\right)\to 0 
\end{align*}
as $N \to \infty$, which implies condensation cannot occur in the sense of Definition \ref{condensation} contradicting the initial assumption, therefore $\phi_{c}< \infty$.

Step (2), prove $Z_{L,N}/w(N)$ converges as $N \to \infty$. By Definition \ref{condensation} the limit
\begin{equation}
a_{K}:=\lim_{N \to \infty}\pi_{L,N}[M_{L}\geq N-K]\ ,
\end{equation}
exists and $a_{K}>0$ for $K$ sufficiently large. For $N > K$ we have
\begin{equation}
\label{eqn:forconvg}
\pi_{L,N}[M_{L}\geq N-K] = L\frac{w(N)}{Z_{L,N}}\sum_{n=0}^{K}Z_{L-1,n}\frac{w(N-n)}{w(N)} \ .
\end{equation}
Since $w(N-n)/w(N) \to \phi_{c}^{n}$, $K$ is fixed, and $a_{K}>0$, \eqref{eqn:forconvg} implies the convergence of $Z_{L,N}/w(N)$ as $N \to \infty$. 

Step (3), prove $z(\phi_{c})< \infty$. By \eqref{eqn:cond1} we have $a_{K} \to 1$ as $K \to \infty$, taking the limit as $N \to \infty$ of \eqref{eqn:forconvg} this implies 
\begin{equation}
\lim_{K \to \infty}\sum_{n=0}^{K}Z_{L-1,n}\phi_{c}^{n} = \sum_{n=0}^{\infty}Z_{L-1,n}\phi_{c}^{n}  < \infty \ .
\end{equation}
Since we also have $\sum_{n=0}^{\infty}Z_{L-1,n}\phi_{c}^{n} = z(\phi_{c})^{L-1}$, this implies $z(\phi_{c}) < \infty$. Using $a_{K} \to 1$, \eqref{eqn:forconvg} then also implies $Z_{L,N}/w(N) \to L z(\phi_{c})^{L-1}$ as $N \to \infty$.

Step (4).
We have seen above that condensation implies $\phi_{c}< \infty$, $z(\phi_{c})<\infty$, and $Z_{L,N}/w(N) \to L z(\phi_{c})^{L-1}$ as $N \to \infty$, then \cite[Theorem 2.10]{embrechts1982convolution} implies
\begin{equation}
\nonumber
 \lim_{N \to \infty}\frac{Z_{2,N}}{w(N)} = 2 z(\phi_{c}) \ ,
\end{equation}
completing this part of the proof.

Now, let us consider a stochastic particle system with canonical distributions of the form (\ref{eqn: prod measure}) which fulfil (\ref{wlog}) and (\ref{paff2}) with $\phi_c =1$ and $z(1)<\infty$. We keep the notation for $\phi_c =1$ general in the following to clarify the argument. It is immediate from Proposition \ref{2tol}, and remembering that we set $w(0)=1$, that
\begin{equation*}
\pi_{L,N} [M_L =N]=Lw(N)/Z_{L,N} \to z(\phi_c )^{-(L-1)} >0\ .
\end{equation*}
Then we have for all fixed $K$ and $N>K$
\begin{align*}
\pi_{L,N} [M_L \geq N-K]&=L\sum_{n=0}^K \frac{w(N-n)Z_{L-1,n}}{Z_{L,N}} =\sum_{n=0}^K Z_{L-1,n}\frac{w(N-n)}{w(N)}\frac{Lw(N)}{Z_{L,N}}\nonumber\\
&\to \sum_{n=0}^K \frac{Z_{L-1,n} \phi_c^n }{z(\phi_c )^{L-1}} =\nu_{\phi_c} (\eta_1 +\ldots +\eta_{L-1} \leq K)
\end{align*}
as $N\to\infty$. Since $\nu_{\phi_c}$ is a non-degenerate probability distribution, this implies that $\nu_{\phi_c} (\eta_1 +\ldots +\eta_{L-1} \leq K)\to 1$ as $K\to\infty$, which is (\ref{eqn:cond1}). 

To compute the distribution outside the maximum we get for fixed $n_1 ,\ldots ,n_{L-1}$ and large enough $N$
\begin{eqnarray}
\lefteqn{\pi_{L,N} [\eta_1 {=}n_1 ,\ldots ,\eta_{L-1}{=}n_{L-1} |M_L {=}\eta_L] =\frac{w(n_1 )\cdots w(n_{L-1})w(N-n_{1} {-}\ldots {-}n_{L-1})}{\pi_{L,N} [M_L {=}\eta_L]\ Z_{L,N}} }\nonumber\\
& &=\frac{1}{L\pi_{L,N} [M_L =\eta_L]}w(n_{1} )\cdots w(n_{L-1})\frac{w(N-n_{1} -\ldots -n_{L-1})}{w(N)}\frac{Lw(N)}{Z_{L,N}}\nonumber\\
& &\to w(n_1 )\cdots w(n_{L-1})\phi_c^{n_{1} +\ldots +n_{L-1}} /z(\phi_c )^{L-1}\ ,
\end{eqnarray}
as $N\to\infty$. Here we have used that spatial homogeneity of the measure and asymptotic uniqueness of the maximum according to \eqref{eqn:cond1} imply $\pi_{L,N} [M_L =\eta_L]\to 1/L$. This completes the proof of Proposition \ref{prop1}.\\

\section{Examples of homogeneous condensing processes}\label{examples}
In this section we review several stochastic particle systems that exhibit condensation. By Theorem \ref{mainres}, if these processes are homogeneous and monotone with a finite critical density they do not have stationary product measures. 
To prove monotonicity for the examples mentioned below it is sufficient to construct a basic coupling of the stochastic process which preserves the partial order and particles jump together with maximal rate.
For a definition of a coupling see \cite{peresbook} and for the statement of Strassen's theorem linking stochastic monotonicity and the coupling technique see \cite{Grimmett2001}. 
The steps to construct a basic coupling are outlined in \cite{saada}. \\

\subsection{Misanthrope processes and generalizations}
\label{sec:mis}
Condensation in homogeneous particle systems has mostly been studied in the framework of misanthrope processes \cite{misanthrope,saada}. At most one particle is allowed to jump at a time and the rate that this occurs depends on the number of particles in the exit and entry sites. The misanthrope process is a stochastic particle system on the state space $\O_{L} = \bbN^{\Lambda}$ defined by the generator
\begin{equation}
\mathcal{L}^{mis}f(\eta) = \sum_{x,y \in \Lambda}{r(\eta_{x},\eta_{y})p(x,y)\big( f(\eta^{x,y}) - f(\eta) \big)} \ .
\label{mis gene}
\end{equation}
Here $\eta^{x,y}=\eta - \delta_{x}+\delta_{y}$ denotes the configuration after a single particle has jumped from site $x$ to site $y$, which occurs with rate $r(\eta_{x},\eta_{y})$. The purely spatial part of the jump rates, $p(x,y)\geq 0 $, are transition probabilities of a random walk on $\Lambda$. Such models are usually studied in a translation invariant setting with periodic boundary conditions, typical choices are symmetric, totally asymmetric or fully connected jump rates with $p(x,y) = 1/2(\delta_{y,x+1}+\delta_{y,x-1})$, $p(x,y) = \delta_{y,x+1}$, or $p(x,y) = (1-\delta_{y,x})/(L-1)$, respectively.

Misanthrope processes include many well-known examples of interacting particle systems, such as zero-range processes \cite{spitzer70}, the inclusion process \cite{giardinaetal09,giardinaetal10}, and the explosive condensation model \cite{waclawetal11}. It is known \cite{misanthrope,Fajfrova2014} that misanthrope processes with translation invariant dynamics $p(x,y) = q(x-y)$ exhibit stationary product measures if and only if the rates fulfil
\begin{equation}
\frac{r(n,m)}{r(m+1,n-1)} = \frac{r(n,0)r(1,m)}{r(m+1,0)r(1,n-1)}\textrm{\quad for all\quad}n\geq 1, m\geq 0\ ,
\end{equation}
and, in addition, either 
\begin{equation}
\begin{cases}
q(z) = q(-z) \textrm{ for all } z \in \Lambda \textrm{ or,} \\
r(n,m)-r(m,n) = r(n,0) - r(m,0) \textrm{ for all } n,m \geq 0 \ .
\end{cases}
\end{equation}
The corresponding stationary weights satisfy
\begin{align}
\frac{w(k+1)}{w(k)} = \frac{w(1)}{w(0)}\frac{r(1,k)}{r(k+1,0)} 
\quad\textrm{and}\quad w(n) =\prod_{k=1}^{n}\frac{r(1,k-1)}{r(k,0)} \ .
\label{misanthrope stationary}
\end{align}
Misanthrope processes are monotone (attractive) \cite{misanthrope} if and only if the jump rates satisfy
\begin{align}
\nonumber
r(n,m) \leq r(n+1,m) &\textrm{ \textit{i.e.} non-decreasing in $n$,}\\
\label{misanthrope monotonicity}
r(n,m) \geq r(n,m+1) &\textrm{ \textit{i.e.} non-increasing in m} \ .
\end{align}

In Theorem \ref{mainres} we have proved that processes that exhibit stationary product measures and condensation with finite mean or power law tails, $w(n) \sim n^{-b}$, with $b \in (3/2,2]$ are necessarily not monotone.  For power law tails with $b \in (1,3/2]$ convergence of $Z_{L,N}/w(N)$ is from below and our method does not disprove monotonicity of the measures $\pi_{L,N}$ or monotonicity of the underlying process. Using the specific form of the stationary measures \eqref{misanthrope stationary}, it is clear that possible examples of monotone processes with stationary product measures of this form cannot be of misanthrope type.
\begin{lemma}
\label{lemma:mismonotone}
A misanthrope process defined by the generator \eqref{mis gene}, that has stationary product measures and exhibits condensation is not monotone.
\end{lemma}
\begin{proof}
\eqref{misanthrope monotonicity} gives necessary conditions for the monotonicity of the misanthrope process and implies with \eqref{misanthrope stationary} that
\begin{equation}
\frac{w(n-1)}{w(n)} = \frac{r(n,0)}{r(1,n-1)}
\end{equation}
is non-decreasing. This implies that the ratio converges to $\phi_c \in (0 , \infty]$, which is the regularity assumption \eqref{regu}. Assuming the process condenses in the sense of Definition \ref{condensation}, then Proposition \ref{prop1} implies $\phi_{c}< \infty$. Now we have
\begin{equation}
\frac{w(n-1)}{w(n)} \leq \phi_{c}\quad \implies  \quad w(n) \geq w(n-1) \phi_{c}^{-1} 
\end{equation}
for all $n \in \bbN$. Therefore, $w(n) \geq w(0)\phi_{c}^{-n}$ which implies
\begin{equation}
\sum_{n=0}^{N}{w(n)\phi_{c}^{n}}\geq w(0)\sum_{n=0}^{N}{\phi_{c}^{n}\phi_{c}^{-n}} \to \infty \textrm{ as }N \to \infty \ .
\end{equation}  
We conclude that the critical partition function diverges and the critical measure $\nu_{\phi_c}$ does not exist, which is a necessary condition for condensation. Therefore condensation does not occur in misanthrope processes with stationary product measures.
\end{proof}

In \cite{saada} generalised misanthrope processes have been introduced where more than one particle is allowed to jump simultaneously. They are defined via transitions $\eta \to \eta + n(\delta_{y}-\delta_{x})$ for $n \in \{0,\ldots , \eta_{x}\}$ at rate $\Gamma_{\eta_{x},\eta_{y}}^{n}(y-x)$ and conditions on the jump rates for monotonicity are characterized. This class provides candidates for possible monotone, condensing processes with product measures as we discuss in the next subsection.

\subsection{Generalised zero-range processes}
The generalised zero-range process (gZRP) \cite{saada} is a stochastic particle system 
on the state space $\Omega_{L} = \mathbb{N}^{\Lambda}$ defined by the generator
\begin{equation}
\label{eq:gzrpgene}
\mathcal{L}^{gZRP}f(\eta) = \sum_{x,y \in \L}\sum_{k=1}^{\eta_{x}}{\alpha_{k}(\eta_{x})p(x,y)\big(f(\eta^{x \to(k) y})-f(\eta)\big)} \ .
\end{equation}
Here $\eta^{x \to(k) y} \in \O_{L}$ is the configuration after $k$ particles have jumped from $x$ to $y \in \Lambda$. The jump rates $\a_{k}(n)$ satisfy $\a_{k}(n) = 0$ if $k>n$, and we use the convention that empty summations are zero. We consider translation invariant $p(x,y)$ on a finite lattice $\L = \{1,\ldots , L\}$ and note that the process preserves particle number $\sum_{x}{\eta_{x}}= N$.

It is known \cite{Evans2004,Fajfrova2014} that these processes exhibit stationary product measures if and only if the jump rates have the explicit form
\begin{equation}
\a_{k}(n) = g(k) \frac{h(n-k)}{h(n)} \ ,
\label{eq:gzrprates}
\end{equation}
where $g,h:\bbN \to [0,\infty )$ are arbitrary non-negative functions  with $h$ strictly positive. The stationary weights are then given by $w(n)=h(n)$. 
Monotonicity of the gZRP can be characterized in terms of
\begin{align}
\label{eq:gzrpR}
R_{k}(n) := \sum_{m=0}^{n-k}{\left(\a_{n-m}(n) - \a_{n+1-m}(n+1)\right)} \ .
\end{align}
The gZRP is monotone if and only if
\begin{align}
\nonumber
R_{k}(n) \geq 0 &\textrm{ for all }n\geq 1 \textrm{ and }k \in \{ 1,\ldots ,n\} \\
\a_{k}(n+1) \geq R_{k}(n) &\textrm{ for all }n\geq 1 \textrm{ and } k \in \{1,\ldots ,n\} \ .
\label{eq:gzrpmono}
\end{align}
We note these conditions arise from a special case of the results in \cite[Theorem 2.11]{saada} on generalised misanthrope models, since $\a_{k}(n)$ depends only on the occupation of the exit site and not the entry site.

In this class, which is also discussed in detail in \cite{Fajfrova2014}, condensing processes which are monotone, homogeneous, and have stationary product measures with a power tail $w(n) \sim n^{-b}$ with $b \in (1,3/2]$ are conjectured to exist. As an example, consider the gZRP with rates given by
\begin{equation}\label{gzrpex}
\a_{k}(n) = \begin{cases}
0 &\textrm{ if }k=0 \textrm{ or }n=0  \\
k^{-b}(1-\frac{k}{n})^{-b} &\textrm{ if } k \in \{1,\ldots ,n-1 \} \\
1 &\textrm{ otherwise} \ .
\end{cases}
\end{equation}
Since $\a_{k}(n)$ is of the form \eqref{eq:gzrprates} the process exhibits stationary product measures with weights of the form
\begin{align*}
w(n) = \begin{cases} 1 &\textrm{ if } n=0 \\
n^{-b} &\textrm{ otherwise } \end{cases}\ .
\end{align*}
For all $b>1$ and $L \geq 2$ the ratio $\frac{Z_{L,N}}{w(N)} $ converges to $ L z(1)^{L-1}$ as $N \to \infty$ \cite{ferrarietal07} so by Proposition \ref{prop1} the process exhibits condensation. 
To prove the process is monotone we must show the rates satisfy the conditions given in equation \eqref{eq:gzrpmono}. We first prove $R_{k}(n) \geq 0 $ for all $k \in \{1,\ldots , n \}$ and $n\geq 1$. Since $\a_{n} (n) -\a_{n+1}(n)= 0$ for all $n\geq1$ we can drop the $m=0$ term from the definition of $R_{k}(n)$. We have
\begin{align}
\nonumber
R_{k}(n) = \sum_{m=1}^{n-k}m^{-b}\bigg[ \bigg(1-\frac{m}{n}\bigg)^{-b} - \bigg(1-\frac{m}{n+1}\bigg)^{-b}  \bigg] \ .
\end{align}
Since $(1-x)^{-b}$ is increasing in $x$ for $b>0$ and $m/n > m/(n+1)$ we have
\begin{equation}
\nonumber
R_{k}(n) > 0 \textrm{ for all } k \in \{1,\ldots ,n\} \textrm{ and } n \geq 1 \ .
\end{equation}
We also need to show $\a_{k}(n+1) \geq R_{k}(n)$ for all $k \in \{1,\ldots ,n\}$ and $n \geq 1$. 
Taking discrete derivatives in $k$
\begin{align}
\nonumber
\a_{k+1}(n+1) - R_{k+1}(n)-(\a_{k}(n+1) - R_{k}(n)) =\a_{k}(n) - \a_{k}(n+1)\\
\nonumber
 = k^{-b}\bigg(1-\frac{k}{n}\bigg)^{-b} - k^{-b} \bigg(1 - \frac{k}{n+1}\bigg)^{-b} > 0 \ ,
\end{align}
so $\a_{k}(n+1) - R_{k}(n)$ is an increasing function in $k$. Therefore,
\begin{equation}
\nonumber
\a_{k}(n+1) - R_{k}(n) \geq  \a_{1}(n+1) - R_{1}(n)
\end{equation}
for all $k \geq 1$, and it suffices to show
\begin{equation}
\label{eq:gzrpcond}
A(n):=\a_{1}(n+1) - R_{1}(n) \geq 0 \textrm{ for all } n\geq 1 \ .
\end{equation}
We present numerical evidence in Figure \ref{gzrpnumerics} which corroborates our claim that the process with rates (\ref{gzrpex}) is indeed monotone for $b\in (1,3/2]$ and is not for $b>3/2$.
\begin{figure}[t]
\centering
\includegraphics[scale=1.05,clip=true,trim= 2em 0 0em 0]{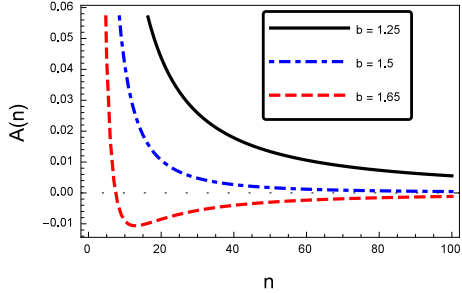}
\caption{\label{gzrpnumerics}
Monotonicity condition \eqref{eq:gzrpcond} for $b=1.25$, $b=1.5$ and $b=1.65$. For $b = 1.65$ the function $A(n)$ falls below zero, implying the gZRP with rates \eqref{gzrpex} is non-monotone. For $b=1.25$ and $b= 1.5$ the function $A(n)$ is positive indicating the process is monotone.}
\end{figure}

\subsection{Homogeneous monotone processes without product measures\label{homnonpro}}
The chipping model is a stochastic particle system on the state space $\Omega_{L}= \mathbb{N}^{\Lambda}$, introduced in \cite{Majumdarb,Majumdar1998}. The dynamics are defined by the generator 
\begin{align}
\nonumber
\mathcal{L}^{chip}f(\eta) =& \sum_{x,y \in \Lambda_{L}}{w\1(\eta_{x}>0)p(x,y) \big(f(\eta^{x,y})-f(\eta) \big)} \\
\label{eq:chipgen}
&+ \sum_{x,y \in \Lambda_{L}}{\1(\eta_{x}>0) p(x,y) \big(f(\eta+\eta_{x}(\delta_{y}-\delta_{x}))-f(\eta) \big)} \ .
\end{align}
Here $\eta + \eta_{x}(\delta_{y}-\delta_{x})$ denotes the configuration after all the particles at site $x$ have jumped to site $y$, which occurs at rate $1$, and single particles jump at rate $w>0$. The spatial part $p(x,y)$ is again spatially homogeneous as described in Section \ref{sec:mis}.

It is easy to see that a basic coupling will preserve the partial order on the state space $\Omega_{L}$ as defined in Section \ref{sec:mon}. Therefore, by Strassen's theorem \cite{Grimmett2001}, the chipping model is a monotone process and Lemma \ref{order} implies that conditional stationary measures of the process are ordered in $N$. The condensation transition in the chipping model was established on a heuristic level in \cite{Majumdarb,Majumdar1998,Rajesh2001}. We have defined the critical density $\rho_{c}$ only for systems with product stationary measures (see \eqref{fmean}). In general, the critical density on a fixed system of size $L\geq 2$, with unique invariant measures $\mu_{L,N}$, can be defined as the background density of bulk sites
\begin{equation}
\rho_{c}(L) := \limsup_{N\to \infty}\frac{\mu_{L,N}\left(N-M_{L}\right)}{L-1} \ .
\end{equation} 
Notice if $\mu_{L,N}$ are conditional product measures (see \eqref{eqn: prod measure}) then $\rho_{c}(L)$ is consistent with \eqref{fmean} and in particular independent of $L$, which follows from Proposition \ref{prop1} (more explicitly \eqref{weaklim}). For the chipping model in the case $L=2$, the process reduces to a $1$-dimensional process on $\{0,\ldots , N \}$ and the measure $\mu_{2,N}$ and $\rho_{c}(2)$ can be computed explicitly to find
\begin{equation}
\rho_{c}(2) = \frac{1}{2}\left( \sqrt{2w+1}-1\right) \ .
\end{equation}
In \cite{Majumdarb,Majumdar1998,Rajesh2001} the critical density in the thermodynamic limit is defined as 
\begin{equation}
\nonumber
\rho_{c}:= \sup \left\{\rho\geq 0\, :\, \frac{\mu_{L,N}(\eta_{x}^{2})}{L} \to 0 \textrm{ as } N,L \to \infty \textrm{ such that }\frac{N}{L} \to \rho \right\} \ ,
\end{equation}
inspired by the fact that in case of condensation the second moment is either dominated by the condensate and scales like the system size $L$, or it diverges since the maximal invariant measure does not have finite second moment. It is shown by heuristic computations in a mean-field limit that 
\begin{equation}
\nonumber
\rho_{c} =  \sqrt{w+1}-1 \ .
\end{equation}
This suggests that the critical density can depend on the system size $L$ for distributions with non-product stationary measures.

 The $\sqrt{w}$ scaling of the critical density can be intuitively understood in the two site chipping model with $N$ particles. This process can be interpreted as a symmetric random walk on the state space $\{0,\ldots,N\}$ with jumps $i \to  i \pm 1$ at rate $w$ and random jumps to either boundary (resetting, $i \to 0$ or $N$) at rate $1$. After a reset the particle diffuses at rate $w$ and reaches a typical distance $\sqrt{w}$ from the boundary until the next reset. So this model is a monotone and spatially homogeneous process that heuristically exhibits a condensation transition with finite (size dependent) critical density, but it does not exhibit stationary product measures. Condensation is also observed in models where chipping is absent ($w=0$) and the dynamics result in a single block of particles jumping on the lattice $\{1,\ldots ,L \}$ corresponding to the critical density $\rho_c =0$. 

\appendix
\section{Connection to statistical mechanics\label{sec:statmech}}

Condensation and non-monotonicity are also related to convexity properties of the entropy, which we briefly describe in the following in a non self-contained and non-rigorous discussion that is aimed at readers with a background in statistical mechanics. In the thermodynamic limit the canonical entropy is defined as
\begin{equation}
s(\rho ):=\lim_{\substack{L\to\infty \\ N/L\to\rho}} \frac{1}{L}\log Z_{L,N}\ .
\label{eq:entro}
\end{equation}
For the processes we consider, equivalence of canonical and grand-canonical ensembles has been established in \cite{stefan} for condensing or non-condensing systems, so $s(\rho )$ is given by the (logarithmic) Legendre transform of the pressure
\begin{equation}
p(\phi ):=\log z(\phi )\ .
\end{equation}
This takes a particularly simple form since the grand-canonical measures are factorisable, and is a strictly convex function for $\phi <\phi_c$. General results then imply that $s(\rho)$ also has to be strictly convex below the critical density $\rho_c$ (see e.g.\ \cite{touchette2009}), which holds for non-condensing systems and condensing systems with $\rho_{c}=\infty$. 
For condensing systems with finite critical density $s(\rho )$ is linear for $\rho >\rho_c$, consistent with phase separation phenomena, where in this case the condensed phase formally exhibits density $\infty$ (see e.g. \cite{Chleboun2015} for a general discussion). 

It is not possible to derive general results for finite $L$ and $N$, but if we assume that the ratio of weights $w(n-1)/w(n)$ is monotone increasing in $n$, we can show that a monotone order of $\pi_{L,N}$ implies that $N\mapsto \frac{1}{L}\log Z_{L,N}$ is necessarily convex. 
Note that with (\ref{regu}) our assumption implies that $w(n)$ has exponential tails with $\phi_c \in (0,\infty )$ or decays super-exponentially with $\phi_c =\infty$, and in both cases the system does not exhibit condensation. 
We can define $w(-1)=0$ so that $w(\eta_x -1)/w(\eta_x )$ is a monotone increasing test function on $\Omega_L$. It is easy to see that for its canonical expectation we have for all $L\geq 2$ and $N\geq 2$
\begin{equation}
\pi_{L,N} \Big(\frac{w(\eta_x -1)}{w(\eta_x )}\Big) =\frac{Z_{L,N-1}}{Z_{L,N}}\ .
\label{eq:exe}
\end{equation}
Therefore, monotonicity of the canonical measures implies that the ratio of partition functions (\ref{eq:exe}) is increasing and the discrete derivative of $\log Z_{L,N}$ in $N$ is decreasing. We expect that in the limit $L\to \infty$ the monotonicity assumption on $w(n-1)/w(n)$ is not necessary, and $\frac{1}{L}\log Z_{L,N}$ is convex in $N$ for all non-condensing systems, consistent with strict convexity of $s(\rho )$. 


For condensing systems the weights $w$ decay sub-exponentially, and if $w(n-1)/w(n)$ is monotone then it has to be decreasing in $n$. Therefore the choice $w(-1)=0$ implies $f(\eta) = w(\eta_{x}-1)/w(\eta_{x})$ is not a monotone test function, and the above general arguments cannot be used to relate non-convexity of $\frac{1}{L}\log Z_{L,N}$ to the absence of a monotone order in $\pi_{L,N}$. 
For particular condensing systems, however, it has been shown that $\frac{1}{L}\log Z_{L,N}$ is typically convex for small $N<\rho_c L$ and concave for larger $N>\rho_c L$ \cite{chlebounetal10,agl}. These results focus on power law and stretched exponential tails for $w(n)$, and have been derived for zero-range processes where the ratio $Z_{L,N-1}/Z_{L,N}$ is equal to the canonical current. Non-monotone behaviour around the critical density therefore has implications for finite-size corrections and derivations of hydrodynamic limits as mentioned in the introduction.
\section{The infinite mean power law case with $b=2$}
\label{bis2}
Consider stationary weights of the form $w(n)= n^{-2}$ with $w(0) = 1$, we prove the non-monotonicity of $Z_{L,N}/w(N)$ in a similar fashion to the proof of Lemma \ref{lemma: infinite mean} summarised in the following lemma.
\begin{lemma}
\label{lemmabis2}
For stationary weights of the form  $w(n)= n^{-2}$ with $w(0) = 1$ we have
\begin{equation}
\nonumber
\frac{N}{\log(N)}\left(\frac{Z_{2,N}}{w(N)} - 2z(1) \right) \to \hat{F}_{2} = 4 \textrm{ as }N \to \infty \ .
\end{equation}
For $L>2$ we have
\begin{equation}
\label{bis2rec}
\frac{N}{\log(N)}\left(\frac{Z_{L,N}}{w(N)} - Lz(1)^{L-1} \right) \to \hat{F}_{L}:=z(1)\hat{F}_{L-1} + (L-1)z(1)^{L-2}\hat{F}_{2} \textrm{ as }N \to \infty \ ,
\end{equation}
which is positive for all $L \geq 2$ since $\hat{F}_{2}>0$.
\end{lemma}
\begin{proof}
First consider the case $L = 2$. As in the proof of Lemma \ref{lemma: infinite mean} we will utilise the full Taylor expansion of $(1-x)^{-2}$, integral bounds on monotone series, and assume $N$ is even, for $N$ odd there exists obvious modifications to the proof. We have
\begin{align*}
\nonumber
\frac{Z_{2,N}}{w(N)} - 2z(1) &= 2\sum_{n=1}^{N/2}n^{-2}\left(1-\frac{n}{N}\right)^{-2} - 2 \sum_{n=1}^{\infty}n^{-2} -2^{4}N^{-2} \ .
\end{align*}
Where the terms $n=0$ in the above summations cancel. Substituting the Taylor expansion of $(1-x)^{-2}$ we find
\begin{align}
\label{eq:pl2}
\frac{Z_{2,N}}{w(N)} - 2z(1) =2\sum_{i=1}^{\infty}{(i+1)N^{-i}\sum_{n=1}^{N/2}{n^{-2+i}}} - 2 \sum_{n=N/2+1}^{\infty}n^{-2} -2^{4}N^{-2} \ .
\end{align}
Now we are in a position to apply the integral bounds \eqref{eq:intbounds}, first noting that $n^{-2+i}$ is decreasing for $i=1$, constant and equal to $1$ for $i=2$, and increasing for $i\geq 3$. Multiplying both sides of \eqref{eq:pl2} and applying the integral bounds it is easy to show
\begin{equation}
\nonumber
\frac{N}{\log(N)}\left(\frac{Z_{2,N}}{w(N)}-2z(1)\right) \to 4 \textrm{ as }N \to \infty \ .
\end{equation}

Now consider the case $L >2$ and make the following inductive hypothesis
\begin{equation}
\nonumber
\lim_{N\to \infty}\frac{N}{\log(N)}\left( \frac{Z_{L,N}}{w(N)} - Lz(1)^{L-1} \right)  = \hat{F}_{L} = z(1)\hat{F}_{L-1} +  (L-1) z(1)^{L-2}\hat{F}_{2} \ .
\end{equation}

As in the proof of Lemma \ref{lemma: infinite mean} write
\begin{align}
\nonumber
\frac{N}{\log(N)}\bigg(\frac{Z_{L+1,N}}{w(N)} &- (L+1)z(1)^{L} \bigg)\\
= \underbrace{\frac{N}{\log(N)}\bigg(\sum_{n=0}^{N/2}{Z_{L,n}\frac{w(N-n)}{w(N)}}-z(1)^{L}\bigg)}_{\hat{\Xi}_{L,N}} +& \underbrace{\frac{N}{\log(N)}\bigg(\sum_{n=0}^{N/2-1}{w(n)\frac{Z_{L,N-n}}{w(N)}}-Lz(1)^{L} \bigg)}_{\hat{\Theta}_{L,N}}\ .
\label{chuck2}
\end{align}
We fist establish the limit of $\hat{\Theta}_{L,N}$ in \eqref{chuck2}. 
The inductive hypothesis can be rewritten as
\begin{equation}
\label{eq:induction3}
\frac{Z_{L,n}}{w(n)} = \left(\hat{F}_{L}+o_{N}(1)\right)\frac{\log(N)}{N}+L z(1)^{L-1}
\end{equation}
Similar to the proof of Lemma \ref{lemma: infinite mean} $\hat{\Theta}_{L,N}$ can be written in the form
\begin{align}
\nonumber
\hat{\Theta}_{L,N} =& \frac{N}{\log(N)}\left( \hat{F}_{L}+o_{N}(1)\right) \left(\sum_{n=0}^{N/2-1}w(n)\frac{w(N-n)}{w(N)}\frac{\log(N-n)}{N-n} \right) \\
\label{eq:hatTheta}
& +Lz(1)^{L-1}\frac{N}{\log(N)}\left(\sum_{n=0}^{N/2-1}{w(n)\frac{w(N-n)}{w(N)}} - z(1) \right) \ .
\end{align}
Since $\log(N-n)$ is decreasing for $n \in \{0,\ldots N/2-1\}$ we can find upper and lower bounds of the first term, by pulling out the logarithm, of the form
\begin{align*}
 \frac{\log(N/2-1)}{\log(N)}\left( \hat{F}_{L}+o_{N}(1)\right) \left(\sum_{n=0}^{N/2-1}w(n)\frac{w(N-n)}{w(N)}\frac{N}{N-n} \right)
\\
\leq \frac{N}{\log(N)}\left( \hat{F}_{L}+o_{N}(1)\right) \left(\sum_{n=0}^{N/2-1}w(n)\frac{w(N-n)}{w(N)}\frac{\log(N-n)}{N-n} \right) \\
\leq\left( \hat{F}_{L}+o_{N}(1)\right) \left(\sum_{n=0}^{N/2-1}w(n)\frac{w(N-n)}{w(N)}\frac{N}{N-n} \right) \ .
\end{align*}
Applying the same steps exactly as they appear in the proof of Lemma \ref{lemma: infinite mean} to the upper and lower bounds we have
\begin{equation}
\label{eq:ThetaLimit}
\lim_{N\to \infty}\hat{\Theta}_{L,N} = z(1)\hat{F_{L}} + \frac{1}{2} Lz(1)^{L-1}\hat{F}_{2} \ .
\end{equation}

To identify the limit of $\hat{\Xi}_{L,N}$ in \eqref{chuck2} we again follow the steps given in the proof of Lemma \ref{lemma: infinite mean}, which implies
\begin{equation}
\lim_{N\to \infty}\hat{\Xi}_{L,N} = \frac{1}{2}Lz(1)^{L-1}\hat{F}_{2} \ .
\end{equation}
Combining this with \eqref{eq:ThetaLimit} we have
\begin{equation}
\nonumber
\frac{N}{\log(N)}\left(\frac{Z_{L+1,N}}{w(N)} - (L+1)z(1)^{L}\right) \to \hat{F}_{L+1} = z(1)\hat{F}_{L} + L z(1)^{L-1}\hat{F}_{2} \,\textrm{ as }N \to \infty\ .
\end{equation}
From the recursion \eqref{bis2rec} it is obvious that $\hat{F}_{L}$ will have the same sign as $\hat{F}_{2}$, completing the proof of Lemma \ref{lemmabis2}.

\end{proof}

\section{On the sign of $F_{2}(b)$}
\label{appendixSign}

In this section, we compute the sign of $F_{2}(b)$ for $b \in (1,2)$, where
\begin{equation}
\nonumber
F_{2}(b) =2\sum_{i=1}^{\infty}\frac{1}{i!}\prod_{j=0}^{i-1}(j+b)\frac{2^{b-1-i}}{1-b+i} - 2\frac{2^{b-1}}{(b-1)} \ .
\end{equation}

Recall the definition of the Pochhammer symbol
\begin{equation}
\nonumber
(q)_{n} = \begin{cases}
1 &\textrm{ if } n=0 \\
(q)(q+1)\ldots (q+n-2)(q+n-1) &\textrm{ for }n\geq 1 
\end{cases} \ ,
\end{equation}
and the hypergeometric function
\begin{equation}
\nonumber
_{2}F_{1}(c,d,e,z)= \sum_{i=0}^{\infty}\frac{z^{i}}{i!}\frac{(c)_{i}(d)_{i}}{(e)_{i}} \ .
\end{equation}

We now show
\begin{equation}
\label{eq:f2bfirst}
F_{2}(b) = - \frac{2^{b}}{b-1} \,{}_{2}F_{1}\left(1-b,b,2-b,\frac{1}{2}\right) \ ,
\end{equation}
which in particular implies $F_{2}(3/2) = 0$ by evaluating the hypergeometric formula.
Factorising the term $2^{b}/(b-1)$ from $F_{2}(b)$ and rearranging terms inside the summation we have
\begin{align*}
\nonumber
F_{2}(b) &= \frac{2^{b}}{b-1}\left( \sum_{i=1}^{\infty}{\frac{1}{i!}\left(\frac{1}{2}\right)^{i}\prod_{j=0}^{i-1}(j+b) \frac{b-1}{1-b+i}} - 1 \right) \\
&= \frac{2^{b}}{b-1} \sum_{i=0}^{\infty}{\frac{1}{i!}\left(\frac{1}{2}\right)^{i}\prod_{j=0}^{i-1}(j+b) \frac{b-1}{1-b+i}} \ .
\end{align*}
Now use the following identities to simplify the terms inside the summation
\begin{equation}
\nonumber
\prod_{j=0}^{i-1}(j+b)=(b)_{i} \quad \textrm{ and }\quad (1-b+i) = (1-b)\frac{(2-b)_{i}}{(1-b)_{i}}\ ,
\end{equation}
which gives the required result \eqref{eq:f2bfirst}.

To complete the proof we use the following two relations for hypergeometric functions, Euler's transform \cite[15.3.3]{Abramowitz1965}
\begin{align*}
{}_{2}F_{1}(c,d,e,z) = (1-z)^{e-d-c} {}_{2}F_{1}(e-c,e-d,e,z) \ ,
\end{align*}	
and Gauss's second summation theorem \cite[15.1.24]{Abramowitz1965}
\begin{align*}
{}_{2}F_{1}\left(c,d,\frac{1}{2}(1+c+d),\frac{1}{2}\right) = \frac{\G\left(\frac{1}{2}\right)\G\left(\frac{1}{2}(1+c+d)\right)}{\G \left(\frac{1}{2}(1+c)\right)\G \left(\frac{1}{2}(1+d)\right)} \ .
\end{align*}
Therefore, 
\begin{align*}
F_{2}(b) &= -\frac{2^{b}}{b-1} {}_{2}F_{1}\left(1-b,b,2-b,\frac{1}{2}\right) \\
&= -\frac{2^{2b -1}}{b-1}  {}_{2}F_{1}\left(1,2-2b,2-b,\frac{1}{2}\right) \\
&= -\frac{\sqrt{\pi } 2^{2 b-1} \Gamma (2-b)}{(b-1) \Gamma
   \left(\frac{3}{2}-b\right)} \ .
\end{align*}

To calculate the sign of $F_{2}(b)$ we first note that the gamma function $\G(x)$ is positive for all $x>0$ and negative in the region $-1<x<0$, which implies
\begin{equation}
\nonumber
F_{2}(b) \begin{cases}
<0 \textrm{ for } b \in (1,3/2) \\
>0 \textrm{ for } b \in (3/2,2) 
\end{cases} \ .
\end{equation}

\section*{Acknowledgements}
This work was supported by the Engineering and Physical Sciences Research Council
(EPSRC), Grant No. EP/101358X/1. P.C. acknowledges
fellowship funding from the University of Warwick, Institute of Advanced
Study. The authors are grateful to Ellen Saada and Thierry Gobron for useful discussions and comments on the manuscript.

\bibliographystyle{unsrt}

\begin{bibdiv}
\begin{biblist}

\bib{spitzer70}{article}{
      author={Spitzer, F.},
       title={{Interaction of Markov processes}},
        date={1970},
     journal={Adv. Math.},
      volume={5},
       pages={246\ndash 290},
}

\bib{misanthrope}{article}{
      author={Cocozza-Thivent, C.},
       title={{Processus des misanthropes}},
        date={1985},
        ISSN={0044-3719},
     journal={Z. Wahrscheinlichkeitstheorie},
      volume={70},
      number={4},
       pages={509\ndash 523},
         url={http://link.springer.com/10.1007/BF00531864},
}

\bib{drouffe98}{article}{
      author={Drouffe, J.-M.},
      author={Godr{\`{e}}che, C.},
      author={Camia, F.},
       title={{A simple stochastic model for the dynamics of condensation}},
        date={1998},
        ISSN={0305-4470},
     journal={J. Phys. A-Math. Gen.},
      volume={31},
      number={1},
       pages={L19},
         url={http://stacks.iop.org/0305-4470/31/i=1/a=003},
}

\bib{Jeon2000}{article}{
      author={Jeon, I.},
      author={March, P.},
      author={Pittel, B.},
       title={{Size of the largest cluster under zero-range invariant
  measures}},
        date={2000},
     journal={Ann. Probab.},
      volume={28},
      number={3},
       pages={1162\ndash 1194},
         url={http://projecteuclid.org/euclid.aop/1019160330},
}

\bib{evans00}{article}{
      author={Evans, M.~R.},
       title={{Phase transitions in one-dimensional nonequilibrium systems}},
        date={2000},
        ISSN={0103-9733},
     journal={Braz. J. Phys.},
      volume={30},
      number={1},
       pages={42\ndash 57},
  url={http://www.scielo.br/scielo.php?script=sci{\_}arttext{\&}pid=S0103-97332000000100005{\&}lng=en{\&}nrm=iso{\&}tlng=en},
}

\bib{Godreche2012}{article}{
      author={Godr{\`{e}}che, C.},
      author={Luck, J.~M.},
       title={{Condensation in the inhomogeneous zero-range process: an
  interplay between interaction and diffusion disorder}},
        date={2012},
        ISSN={1742-5468},
     journal={J. Stat. Mech. Theor. Exp.},
      number={12},
       pages={P12013},
  url={http://iopscience.iop.org/1742-5468/2012/12/P12013/article?fromSearchPage=true},
}

\bib{Chleboun2013a}{article}{
      author={Chleboun, P.},
      author={Grosskinsky, S.},
       title={{Condensation in Stochastic Particle Systems with Stationary
  Product Measures}},
        date={2013},
        ISSN={0022-4715},
     journal={J. Stat. Phys.},
      volume={154},
      number={1-2},
       pages={432\ndash 465},
}

\bib{Evans2014}{article}{
      author={Evans, M.~R.},
      author={Waclaw, B.},
       title={{Condensation in stochastic mass transport models: beyond the
  zero-range process}},
        date={2014},
        ISSN={1751-8113},
     journal={J. Phys. A-Math. Theor.},
      volume={47},
      number={9},
       pages={095001},
         url={http://iopscience.iop.org/1751-8121/47/9/095001/article/},
}

\bib{ferrarietal07}{article}{
      author={Ferrari, P.~A.},
      author={Landim, C.},
      author={Sisko, V.},
       title={{Condensation for a Fixed Number of Independent Random
  Variables}},
        date={2007},
        ISSN={0022-4715},
     journal={J. Stat. Phys.},
      volume={128},
      number={5},
       pages={1153\ndash 1158},
         url={http://link.springer.com/10.1007/s10955-007-9356-3},
}

\bib{stefan}{article}{
      author={Grosskinsky, S.},
      author={Sch{\"{u}}tz, G.~M.},
      author={Spohn, H.},
       title={{Condensation in the Zero Range Process: Stationary and Dynamical
  Properties}},
        date={2003},
        ISSN={0022-4715},
     journal={J. Stat. Phys.},
      volume={113},
      number={3-4},
       pages={389\ndash 410},
         url={http://dx.doi.org/10.1023/A:1026008532442},
}

\bib{armendarizetal09}{article}{
      author={Armend{\'{a}}riz, I.},
      author={Loulakis, M.},
       title={{Thermodynamic limit for the invariant measures in supercritical
  zero range processes}},
        date={2008},
        ISSN={0178-8051},
     journal={Probab. Theory Relat. Fields},
      volume={145},
      number={1-2},
       pages={175\ndash 188},
         url={http://link.springer.com/10.1007/s00440-008-0165-7},
}

\bib{agl}{article}{
      author={Armend{\'{a}}riz, I.},
      author={Grosskinsky, S.},
      author={Loulakis, M.},
       title={{Zero range condensation at criticality}},
        date={2013},
     journal={Stoch. Proc. Appl.},
      volume={123},
      number={9},
       pages={3466\ndash 3496},
         url={http://wrap.warwick.ac.uk/35937/},
}

\bib{CharlesM.Goldie}{incollection}{
      author={Goldie, C.~M.},
      author={Kl{\"{u}}ppelberg, C.},
       title={{Subexponential Distributions}},
        date={1998},
   booktitle={A pract. guid. to heavy tails stat. tech. anal. heavy tailed
  distrib.},
   publisher={Birkh{\"{a}}user, Boston},
       pages={435\ndash 459},
         url={http://citeseerx.ist.psu.edu/viewdoc/summary?doi=10.1.1.28.1546},
}

\bib{Teugels1975}{article}{
      author={Teugels, J.~L.},
       title={{The Class of Subexponential Distributions}},
        date={1975},
        ISSN={0091-1798},
     journal={Ann. Probab.},
      volume={3},
       pages={1000\ndash 1011},
}

\bib{Pitman1980}{article}{
      author={Pitman, E. J.~G.},
       title={{Subexponential distribution functions}},
        date={1980},
     journal={J. Austral. Math. Soc. Ser. A},
      volume={29},
       pages={337\ndash 347},
         url={http://journals.cambridge.org/abstract{\_}S1446788700021340},
}

\bib{J.Chover1973}{article}{
      author={Chover, J.},
      author={Ney, P.},
      author={Wainger, S.~D.},
  title={Functions of probability measures},
  journal={Journal d'analyse math{\'e}matique},
  volume={26},
  number={1},
  pages={255--302},
  year={1973},
  publisher={Springer}
}

\bib{saada}{article}{
      author={Gobron, T.},
      author={Saada, E.},
       title={{Couplings, attractiveness and hydrodynamics for conservative
  particle systems}},
        date={2010},
        ISSN={0246-0203},
     journal={Ann. I. H. Poincare-PR},
      volume={46},
      number={4},
       pages={1132\ndash 1177},
         url={http://cat.inist.fr/?aModele=afficheN{\&}cpsidt=23661253},
}

\bib{krugetal96}{article}{
      author={Krug, J.},
      author={Ferrari, P.~A.},
       title={{Phase Transitions in Driven Diffusive Systems With Random
  Rates}},
        date={1996},
     journal={J. Phys. A-Math. Gen.},
      volume={29},
       pages={L465\ndash L471},
}

\bib{landim96}{article}{
      author={Landim, C.},
       title={{Hydrodynamical limit for space inhomogeneous one-dimensional
  totally asymmetric zero-range processes}},
        date={1996},
        ISSN={2168-894X},
     journal={Ann. Probab.},
      volume={24},
      number={2},
       pages={599\ndash 638},
         url={http://projecteuclid.org/euclid.aop/1039639356},
}

\bib{benjaminietal96}{article}{
      author={Benjamini, I.},
      author={Ferrari, P.~A.},
      author={Landim, C.},
       title={{Asymmetric conservative processes with random rates}},
        date={1996},
        ISSN={0304-4149},
     journal={Stoch. Proc. Appl.},
      volume={61},
      number={2},
       pages={181\ndash 204},
  url={http://www.sciencedirect.com/science/article/pii/0304414995000771},
}

\bib{andjeletal00}{article}{
      author={Andjel, E.~D.},
      author={Ferrari, P.~A.},
      author={Guiol, H.},
      author={Landim, C.},
       title={{Convergence to the maximal invariant measure for a zero-range
  process with random rates}},
        date={2000},
        ISSN={0304-4149},
     journal={Stoch. Proc. Appl.},
      volume={90},
      number={1},
       pages={67\ndash 81},
  url={http://www.sciencedirect.com/science/article/pii/S0304414900000375},
}

\bib{ferrarisisko}{article}{
      author={Ferrari, P.~A.},
      author={Sisko, V.},
       title={{Escape of mass in zero-range processes with random rates}},
        date={2007},
     journal={IMS Lect. notes, Asymptotics Part. Process. Inverse Probl.},
      volume={55},
       pages={108\ndash 120},
}

\bib{Bahadoran2014}{article}{
      author={Bahadoran, C.},
      author={Mountford, T.},
      author={Ravishankar, K.},
      author={Saada, E.},
       title={{Supercritical behavior of asymmetric zero-range process with
  sitewise disorder}},
        date={2014},
       pages={48}
}

\bib{Bahadoran2015}{article}{
      author={Bahadoran, C.},
      author={Mountford, T.},
      author={Ravishankar, K.},
      author={Saada, E.},
       title={{Supercriticality conditions for asymmetric zero-range process
  with sitewise disorder}},
        date={2015},
        ISSN={0103-0752},
     journal={Braz. J. Probab. Stat.},
      volume={29},
      number={2},
       pages={313\ndash 335},
  url={http://www.researchgate.net/publication/268451949{\_}Supercriticality{\_}conditions{\_}for{\_}asymmetric{\_}zero-range{\_}process{\_}with{\_}sitewise{\_}disorder},
}

\bib{chlebounetal10}{article}{
      author={Chleboun, P.},
      author={Grosskinsky, S.},
       title={{Finite Size Effects and Metastability in Zero-Range
  Condensation}},
        date={2010},
        ISSN={0022-4715},
     journal={J. Stat. Phys.},
      volume={140},
      number={5},
       pages={846\ndash 872},
         url={http://link.springer.com/10.1007/s10955-010-0017-6},
}

\bib{Stamatakis2014}{article}{
      author={Stamatakis, M.~G.},
       title={{Hydrodynamic Limit of Mean Zero Condensing Zero Range Processes
  with Sub-Critical Initial Profiles}},
        date={2014},
        ISSN={0022-4715},
     journal={J. Stat. Phys.},
      volume={158},
      number={1},
       pages={87\ndash 104},
         url={http://link.springer.com/10.1007/s10955-014-1113-9},
}

\bib{Fajfrova2014}{article}{
      author={Fajfrova, L.},
      author={Gobron, T.},
      author={Saada, E.},
       title={{Invariant measures for mass migration processes}},
       journal={arXiv:1507.00778 [math.PR]},
}

\bib{Waclaw2007}{article}{
      author={Waclaw, B.},
      author={Bogacz, L.},
      author={Burda, Z.},
      author={Janke, W.},
       title={{Condensation in zero-range processes on inhomogeneous
  networks}},
        date={2007},
     journal={Phys. Rev. E},
      volume={76},
      number={4},
       pages={046114},
}
\bib{foss2011introduction}{book}{
  title={An introduction to heavy-tailed and subexponential distributions},
  author={Foss, S.},
  author={Korshunov, D.}
  author={Zachary, S.}
  year={2011},
  publisher={Springer}
}
\bib{Klppelberg1989}{article}{
      author={Kl{\"{u}}ppelberg, C.},
       title={{Subexponential distributions and characterizations of related
  classes}},
        date={1989},
        ISSN={0178-8051},
     journal={Probab. Theory Relat. Fields},
      volume={82},
      number={2},
       pages={259\ndash 269},
         url={http://link.springer.com/10.1007/BF00354763},
}

\bib{Baltrunas2004}{article}{
      author={Baltrunas, A.},
      author={Kl{\"{u}}ppelberg, C.},
       title={{Subexponential Distributions - Large Deviations with
  Applications to Insurance and Queueing Models}},
        date={2004},
        ISSN={1369-1473},
     journal={Aust N Z J Stat},
      volume={46},
      number={1},
       pages={145\ndash 154},
         url={http://doi.wiley.com/10.1111/j.1467-842X.2004.00320.x},
}

\bib{Chleboun2015}{article}{
      author={Chleboun, P.},
      author={Grosskinsky, S.},
       title={{A dynamical transition and metastability in a size-dependent
  zero-range process}},
        date={2015},
        ISSN={1751-8113},
     journal={J. Phys. A-Math. Theor.},
      volume={48},
      number={5},
       pages={055001},
         url={http://stacks.iop.org/1751-8121/48/i=5/a=055001},
}

\bib{Chistyakov1964}{article}{
      author={Chistyakov, V.~P.},
       title={{A Theorem on Sums of Independent Positive Random Variables and
  Its Applications to Branching Random Processes}},
        date={1964},
        ISSN={0040-585X},
     journal={Theory Probab. Appl.},
      volume={9},
      number={4},
       pages={640\ndash 648},
         url={http://epubs.siam.org/doi/abs/10.1137/1109088},
}

\bib{Embrechts1980}{article}{
      author={Embrechts, P.},
      author={Goldie, CM},
       title={{On closure and factorization properties of subexponential and
  related distributions}},
        date={1980},
     journal={J. Austral. Math. Soc. Ser. A},
      volume={29},
      number={02},
       pages={243\ndash 256},
         url={http://journals.cambridge.org/abstract{\_}S1446788700021224},
}

\bib{Denisov2008}{article}{
      author={Denisov, D.},
      author={Dieker, A.~B.},
      author={Shneer, V.},
       title={{Large deviations for random walks under subexponentiality: The
  big-jump domain}},
        date={2008},
        ISSN={2168-894X},
     journal={Ann. Probab.},
      volume={36},
      number={5},
       pages={1946\ndash 1991},
         url={http://projecteuclid.org/euclid.aop/1221138771},
}

\bib{Armendariz2011}{article}{
      author={Armend{\'{a}}riz, I.},
      author={Loulakis, M.},
       title={{Conditional distribution of heavy tailed random variables on
  large deviations of their sum}},
        date={2011},
        ISSN={03044149},
     journal={Stoch. Proc. Appl.},
      volume={121},
      number={5},
       pages={1138\ndash 1147},
  url={http://www.sciencedirect.com/science/article/pii/S0304414911000238},
}
\bib{embrechts1982convolution}{article}{
  author={Embrechts, P.},
  author={Goldie, C. M.},
  title={On convolution tails},
  journal={Stoch. Proc. Appl.},
  volume={13},
  number={3},
  pages={263--278},
  year={1982},
  publisher={Elsevier}
}
\bib{peresbook}{book}{
      author={Levin, D.~A.},
      author={Peres, Y.},
      author={Wilmer, E.~L.},
       title={{Markov Chains and Mixing Times}},
   publisher={American Mathematical Society},
        date={2009},
}

\bib{Grimmett2001}{book}{
      author={Grimmett, G.},
      author={Stirzaker, D.},
       title={{Probability and Random Processes}},
   publisher={OUP Oxford},
        date={2001},
        ISBN={0198572220},
  url={http://www.amazon.co.uk/Probability-Random-Processes-Geoffrey-Grimmett/dp/0198572220},
}

\bib{giardinaetal09}{article}{
      author={Giardin{\`{a}}, C.},
      author={Kurchan, J.},
      author={Redig, F.},
      author={Vafayi, K.},
       title={{Duality and Hidden Symmetries in Interacting Particle Systems}},
        date={2009},
        ISSN={0022-4715},
     journal={J. Stat. Phys.},
      volume={135},
      number={1},
       pages={25\ndash 55},
         url={http://link.springer.com/10.1007/s10955-009-9716-2},
}

\bib{giardinaetal10}{article}{
      author={Giardin{\`{a}}, C.},
      author={Redig, F.},
      author={Vafayi, K.},
       title={{Correlation Inequalities for Interacting Particle Systems with
  Duality}},
        date={2010},
        ISSN={0022-4715},
     journal={J. Stat. Phys.},
      volume={141},
      number={2},
       pages={242\ndash 263},
         url={http://link.springer.com/10.1007/s10955-010-0055-0},
}

\bib{waclawetal11}{article}{
      author={Waclaw, B.},
      author={Evans, M.~R.},
       title={{Explosive Condensation in a Mass Transport Model}},
        date={2012},
        ISSN={0031-9007},
     journal={Phys. Rev. Lett.},
      volume={108},
      number={7},
       pages={70601},
         url={http://link.aps.org/doi/10.1103/PhysRevLett.108.070601},
}

\bib{Evans2004}{article}{
      author={Evans, M.~R.},
      author={Majumdar, S.~N.},
      author={Zia, R. K.~P.},
       title={{Factorized steady states in mass transport models}},
        date={2004},
        ISSN={0305-4470},
     journal={J. Phys. A-Math. Gen.},
      volume={37},
      number={25},
       pages={L275\ndash L280},
         url={http://stacks.iop.org/0305-4470/37/i=25/a=L02},
}

\bib{Majumdarb}{article}{
      author={Majumdar, S.~N.},
      author={Krishnamurthy, S.},
      author={Barma, M.},
       title={{Nonequilibrium Phase Transition in a Model of Diffusion,
  Aggregation, and Fragmentation}},
        ISSN={1572-9613},
     journal={J. Stat. Phys.},
      volume={99},
      number={1-2},
       pages={1\ndash 29},
         url={http://link.springer.com/article/10.1023/A:1018632005018},
}

\bib{Majumdar1998}{article}{
      author={Majumdar, S.~N.},
      author={Krishnamurthy, S.},
      author={Barma, M.},
       title={{Nonequilibrium Phase Transitions in Models of Aggregation,
  Adsorption, and Dissociation}},
        date={1998},
        ISSN={0031-9007},
     journal={Phys. Rev. Lett.},
      volume={81},
      number={17},
       pages={3691\ndash 3694},
         url={http://link.aps.org/doi/10.1103/PhysRevLett.81.3691},
}

\bib{Rajesh2001}{article}{
      author={Rajesh, R.},
      author={Majumdar, S.~N.},
       title={{Exact phase diagram of a model with aggregation and chipping}},
        date={2001},
        ISSN={1063-651X},
     journal={Phys. Rev. E},
      volume={63},
      number={3},
       pages={036114},
         url={http://link.aps.org/doi/10.1103/PhysRevE.63.036114
  http://www.ncbi.nlm.nih.gov/pubmed/11308716},
}

\bib{touchette2009}{article}{
      author={Touchette, H.},
       title={{The large deviation approach to statistical mechanics}},
        date={2009},
        ISSN={03701573},
     journal={Phys. Rep.},
      volume={478},
      number={1-3},
       pages={1\ndash 69},
}

\bib{georgii1979canonical}{book}{
  title={Canonical Gibbs Measures},
  author={Georgii, Hans-Otto},
  year={1979},
  publisher={Springer}
}


\bib{Abramowitz1965}{book}{
      author={Abramowitz, M.},
      author={Stegun, I.~A.},
       title={{Handbook of Mathematical Functions}},
   publisher={Dover, New York},
        date={1965},
      number={3},
        ISBN={0486612724},
         url={http://www.ncbi.nlm.nih.gov/pubmed/22118217},
}



\end{biblist}
\end{bibdiv}


\end{document}